%% file: main.tex
\documentclass[10 pt, onecolumn]{IEEEtran}
\IEEEoverridecommandlockouts

\usepackage{graphicx}          
\graphicspath{{figures/}}

\usepackage{ifthen}
\newboolean{BGRAPHPDF}
\setboolean{BGRAPHPDF}{true}
\newboolean{TECHREP}
\setboolean{TECHREP}{true}

\ifthenelse {\boolean{BGRAPHPDF}}
{
\usepackage[ruled]{algorithm}
\usepackage{algpseudocode}
}
{
\usepackage{vaucanson-g}
}

\usepackage{color,soul}
\newboolean{HIGHCOMM}
\setboolean{HIGHCOMM}{false}
\newcommand \yhl[1]{\ifthenelse{\boolean{HIGHCOMM}}{\textcolor{blue}{#1}}{#1}}
\newboolean{SHOW_WHAT_IT_WAS}
\setboolean{SHOW_WHAT_IT_WAS}{true}
\newcommand \shl[1]{\ifthenelse{\boolean{SHOW_WHAT_IT_WAS}}{\yhl{\sout{#1}}}{}}
\usepackage{ulem}
\usepackage{tikz}
\newcommand*\mycirc[1]{%
\begin{tikzpicture}
\node[draw,circle,inner sep=1pt] {#1};
\end{tikzpicture}}

\usepackage{natbib}
\usepackage{fainekos-macros}

\title{\LARGE \bf
On the Minimal Revision Problem of Specification Automata
}

\author{Kangjin Kim, Georgios Fainekos and Sriram Sankaranarayanan%
\thanks{This work has been partially supported by award NSF CNS 1116136.}%
\thanks{K. Kim and G. Fainekos are with the School of Computing, Informatics and Decision Systems Engineering, Arizona State University, Tempe, AZ 85281, USA { \tt\small \{Kangjin.Kim,fainekos\}@asu.edu} }
\thanks{S. Sankaranarayanan is with the Department of Computer Science, University of Colorado, Boulder, CO { \tt\small srirams@colorado.edu}}%
}

\begin{document}
\pagenumbering{arabic}

\maketitle
%

\begin{abstract} 
As robots are being integrated into our daily lives, it becomes necessary to provide guarantees on the safe and provably correct operation.
Such guarantees can be provided using automata theoretic task and mission planning where the requirements are expressed as temporal logic specifications.
However, in real-life scenarios, it is to be expected that not all user task requirements can be realized by the robot.
In such cases, the robot must provide feedback to the user on why it cannot accomplish a given task.
Moreover, the robot should indicate what tasks it can accomplish which are as ``close" as possible to the initial user intent. 
This paper establishes that the latter problem, which is referred to as the minimal specification revision problem, is NP complete.
A heuristic algorithm is presented that can compute good approximations to the Minimal Revision Problem (MRP) in polynomial time.
The experimental study of the algorithm demonstrates that in most problem instances the heuristic algorithm actually returns the optimal solution.
Finally, some cases where the algorithm does not return the optimal solution are presented.

\end{abstract}



\input{introduction}

\input{problem}

\input{planning}

\input{automaton_revision_sat}

\input{heuristic}

\input{experiments}

\input{related_work}

\input{conclusions}



\input{appendix}
\input{upperboundproof}

\bibliographystyle{plainnat}
\bibliography{mypapers,fainekos_bibrefs}
\end{document}

%% file: introduction.tex
\section{Introduction}

As robots become mechanically more capable, they are going to be more and more integrated into our daily lives.
Non-expert users will have to communicate with the robots in a natural language setting and request a robot or a team of robots to accomplish complicated tasks.
Therefore, we need methods that can capture the high-level user requirements, solve the planning problem and map the solution to low level continuous control actions.
In addition, such frameworks must come with mathematical guarantees of safe and correct operation for the whole system and not just the high level planning or the low level continuous control.

Linear Temporal Logic (LTL) (see \citealt{ClarkeGP99}) can provide the mathematical framework that can bridge the gap between 
\begin{enumerate}
\item natural language and high-level planning algorithms (e.g., \citealt{KressGazitFP08ar,DzifcakSBS09icra}), and  
\item high-level planning algorithms and control (e.g., \citealt{FainekosGKGP09automatica,KaramanSF08cdc,BhatiaKV10icra,WongpiromsarnTM10hscc,RoyTM11hscc}).
\end{enumerate}

LTL has been utilized as a specification language in a wide range of robotics applications.
For a good coverage of the current research directions, the reader is referred to  \citealt{FainekosGKGP09automatica,KressGazitFP09tro,KaramanSF08cdc,KloetzerB10tro,BhatiaKV10icra,WongpiromsarnTM10hscc,RoyTM11hscc,BobadillaEtAlRSS11,UlusoyEtAl2011iros,LacerdaL2011rss,LaViersEtAl2011iccps,FilippidisDK12cdc} and the references therein.

For instance, in \citealt{FainekosGKGP09automatica}, the authors present a framework for motion planning of a single mobile robot with second order dynamics.
The problem of reactive planning and distributed controller synthesis for multiple robots is presented in \citealt{KressGazitFP09tro} for a fragment of LTL (Generalized Reactivity 1 (GR1)).
The authors in \citealt{WongpiromsarnTM10hscc} present a method for incremental planning when the specifications are provided in the GR1 fragment of LTL.
The papers \citealt{KloetzerB10tro,UlusoyEtAl2011iros} address the problem of centralized control of multiple robots where the specifications are provided as LTL formulas.
An application of LTL planning methods to humanoid robot dancing is presented in \citealt{LaViersEtAl2011iccps}.
In \citealt{KaramanSF08cdc}, the authors convert the LTL planning problem into Mixed Integer Linear Programming (MILP) or Mixed Integer Quadratic Programming (MIQP) problems.
The use of sampling-based methods for solving the LTL motion planning problem is explored in \citealt{BhatiaKV10icra}.
All the previous applications assume that the robots are autonomous agents with full control over their actions.
An interesting different approach is taken in \citealt{BobadillaEtAlRSS11} where the agents move uncontrollably in the environment and the controller opens and closes gates in the environment.

All the previous methods are based on the assumption that the LTL planning problem has a feasible solution.
However, in real-life scenarios, it is to be expected that not all complex task requirements can be realized by a robot or a team of robots.
In such failure cases, the robot needs to provide feedback to the non-expert user on why the specification failed.
Furthermore, it would be desirable that the robot proposes a number of plans that can be realized by the robot and which are as ``close" as possible to the initial user intent.
Then, the user would be able to understand what are the limitations of the robot and, also, he/she would be able to choose among a number of possible feasible plans.

In \citealt{Fainekos11icra}, we made the first steps towards solving the debugging (i.e., why the planning failed) and revision (i.e., what the robot can actually do) problems for automata theoretic LTL planning (\citealt{GiacomoV99ecp}).
We remark that a large number of robotic applications, e.g.,  \citealt{FainekosGKGP09automatica,KloetzerB10tro,BhatiaKV10icra,BobadillaEtAlRSS11,UlusoyEtAl2011iros} and \citealt{LaViersEtAl2011iccps}, are utilizing this particular LTL planning method.

In the follow-up paper \citealt{KimFS12icra}, we studied the theoretical foundations of the specification revision problem when both the system and the specification can be represented by $\omega$-automata (\citealt{Buchi60}).
In particular, we focused on the Minimal Revision Problem (MRP), i.e., finding the ``closest" satisfiable specification to the initial specification, and we proved that the problem is NP-complete even when severely restricting the search space.
Furthermore, we presented an encoding of MRP as a satisfiability problem and we demonstrated experimentally that we can quickly get the exact solution to MRP  for small problem instances. 

In \citealt{KimFS12iros}, we revisited MRP and we presented a heuristic algorithm that can approximately solve MRP in polynomial time.
We experimentally established that the heuristic algorithm almost always returns the optimal solution on random problem instances and on LTL planning scenarios from our previous work.
Furthermore, we demonstrated that we can quickly return a solution to the MRP problem on large problem instances.
Finally, we provided examples where the algorithm is guaranteed not to return the optimal solution.

This paper is an extension of the preliminary work by \citealt{KimFS12icra} and \citealt{KimFS12iros}. 
In this extended journal version, we present a unified view of the theory alongside with the proofs that were omitted in the aforementioned papers.
The class of specifications that can be handled by our framework has been extended as well.
The framework that was presented in \citealt{KimFS12icra} was geared towards robotic motion planning specifications as presented in \citealt{FainekosGKGP09automatica}.
In addition, we prove that the heuristic algorithm that we presented in \citealt{KimFS12iros} has a constant approximation ratio only on a specific class of graphs.
Furthermore, we have included several running examples to enhance the readability of the paper as well as more case studies to demonstrate the feasibility of the framework.
On the other hand, we have excluded some details on the satisfiability encoding of the MRP problem which can be found in \citealt{KimFS12icra}.

%% file: problem.tex
\section{Problem Formulation}
\label{sec:problem}

In this paper, we work with discrete abstractions (Finite State Machines) of the continuous robotic control system (\citealt{FainekosGKGP09automatica}).
This is a common practice in approaches that hierarchically decompose the control synthesis problem into high level discrete planning synthesis and low level continuous feedback controller composition (e.g., \citealt{FainekosGKGP09automatica,KressGazitFP09tro,LaViersEtAl2011iccps,KloetzerB10tro,UlusoyEtAl2011iros}). 
Each state of the Finite State Machine (FSM) $\FTS$ is labeled by a number of symbols from a set $\Pi= \{\pi_0,$ $\pi_1,$ $\dots,$ $\pi_n\}$ that represent regions in the configuration space (see \citealt{LaValle05} or \citealt{ChosetLHKBKT05book}) of the robot or, more generally, actions that can be performed by the robot.
The control requirements for such a system can be posed using specification automata $\Bc$ with B{\"u}chi acceptance conditions (see \citealt{Buchi60}) also known as $\omega$-automata.


The following example presents a scenario for motion planning of a mobile robot.

\begin{figure}[t]
\centering
\includegraphics[width=7cm]{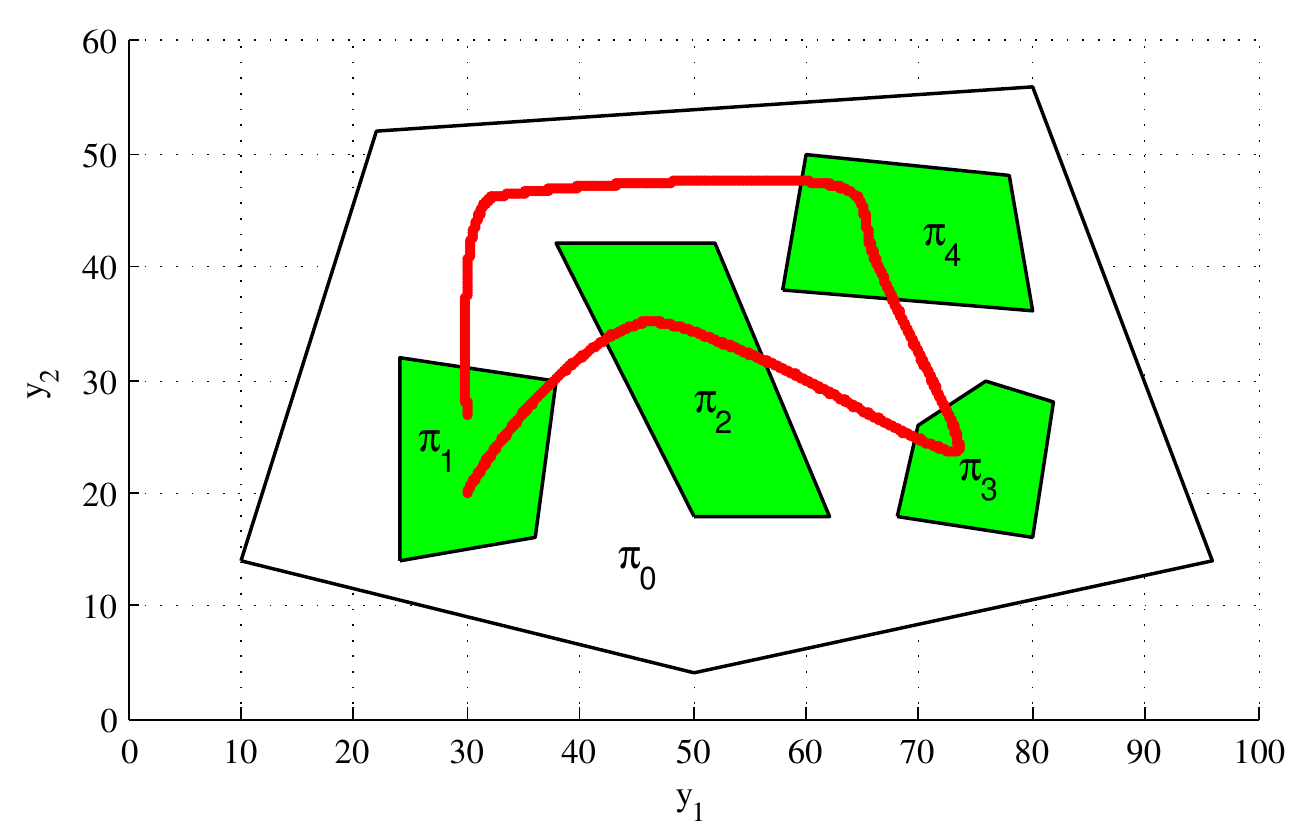}
\caption{The simple environment of Example \ref{exm:areas4} along with a low speed mobile robot trajectory that satisfies the specification.}
\label{pic:areas4}
\end{figure}

\ifthenelse {\boolean{BGRAPHPDF}}
{
\begin{figure}[t]
\centering
\includegraphics[width=8cm]{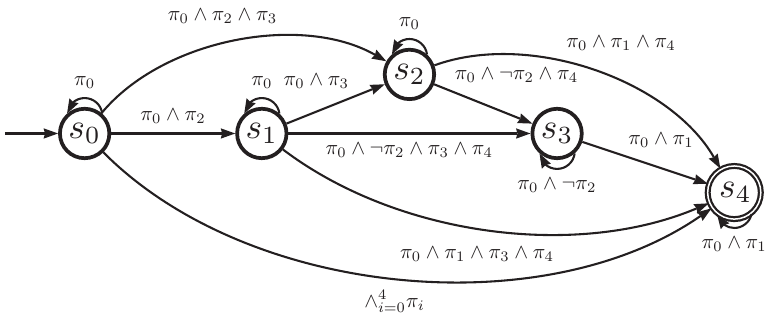}
\caption{The specification automaton $\Bc_{m}$ of Example \ref{exm:areas4}.}
\label{fig:exmp:init:spec}
\end{figure}
}
{
\begin{figure}
\begin{center}
\VCDraw{%
\begin{VCPicture}{(-1,-2.5)(10,2)}
\State[s_0]{(-1,0)}{A} \State[s_2]{(4.5,1)}{B} \State[s_1]{(2,0)}{C}
\State[s_3]{(7,0)}{0} \FinalState[s_4]{(10,-1)}{2}  
\Initial{A}  
\ncline{->}{A}{C} \naput{${\pi}_0 \wedge \pi_2$}
\ncline{->}{C}{B} \naput{${\pi}_0 \wedge \pi_3$}
\ncline{->}{C}{0} \nbput{${\pi}_0 \wedge \neg \pi_2 \wedge \pi_3 \wedge \pi_4$}
\ncline{->}{B}{0} \naput{${\pi}_0 \wedge \neg \pi_2 \wedge \pi_4$}
\ncline{->}{0}{2} \naput{${\pi}_0 \wedge \pi_1$}
\ncarc[arcangle=40]{->}{A}{B} \naput{$\pi_0 \wedge \pi_2 \wedge \pi_3$}
\ncarc[arcangle=-40]{->}{A}{2} \nbput{$\wedge_{i=0}^4 \pi_i$}
\ncarc[arcangle=-30]{->}{C}{2} \nbput{$\pi_0 \wedge \pi_1 \wedge \pi_3 \wedge \pi_4$}
\ncarc[arcangle=40]{->}{B}{2} \naput{$\pi_0 \wedge \pi_1 \wedge \pi_4$}
\nccircle{->}{A}{0.3cm}\nbput{${\pi}_0$}
\nccircle{->}{B}{0.3cm}\nbput{${\pi}_0$}
\nccircle{->}{C}{0.3cm}\nbput{${\pi}_0$}
\nccircle{<-}{0}{-0.3cm}\naput{$\pi_0 \wedge \neg \pi_2$}
\nccircle{<-}{2}{-0.3cm}\naput{${\pi}_0 \wedge \pi_1$}
\end{VCPicture}}
\end{center}
\caption{The specification automaton of Example \ref{exm:areas4}.}
\label{fig:exmp:init:spec}
\end{figure}
}

\begin{exmp}[Robot Motion Planning] 
\label{exm:areas4} 
We consider a mobile robot which operates in a planar environment. 
The continuous state variable $x(t)$ models the internal dynamics of the robot whereas only its position $y(t)$ is observed. 
In this paper, we will consider a 2nd order model of the motion of a planar robot (dynamic model):
\begin{align*}
\dot{x}_1(t)& = x_2(t), &  \; x_1(t)\in \Re^2,\;  x_1(0)\in X_{1,0} \\
\dot{x}_2(t)& = u(t), & \; x_2(t) \in \Re^2, \;   x_2(0)=0, \; u(t) \in U \\
y(t)&= x_1(t). & 
\end{align*}
The robot is moving in a convex polygonal environment $\pi_0$ with four areas of interest denoted by $\pi_1,\pi_2,\pi_3,\pi_4$ (see Fig.~\ref{pic:areas4}). 
Initially, the robot is placed somewhere in the region labeled by $\pi_1$.
The robot must accomplish the task: ``Stay always in $\pi_0$ and visit area $\pi_2$, then area $\pi_3$, then area $\pi_4$ and, finally, return to and stay in region $\pi_1$ while avoiding area $\pi_2$," which is captured by the specification automaton in Fig. \ref{fig:exmp:init:spec}.

In \citealt{FainekosGKGP09automatica}, we developed a hierarchical framework for motion planning for dynamic models of robots.
The hierarchy consists of a high level logic planner that solves the motion planning problem for a kinematic model of the robot, e.g.,
\ifthenelse {\boolean{TECHREP}}
{
\begin{align*}
\dot{z}(t)& = u(t), & \; z(t) \in \Re^2, \;   z(0) \in Z_0 \\
y'(t)&= z(t). & 
\end{align*}
}
{
\[ \dot{z}(t) = u(t), \; y'(t)= z(t), \; z(t) \in \Re^2, \;   z(0) \in Z_0.  \]
}
Then, the resulting hybrid controller is utilized for the design of an approximate tracking controller for the dynamic model.
Since the tracking is approximate, the sets $\pi$ need to be modified (see Fig. \ref{pic:notraj} for an example) depending on the maximum speed of the robot so that the controller has a guaranteed tracking performance.
For example, in Fig. \ref{pic:notraj}, the regions that now must be visited are the contracted yellow regions, while the regions to be avoided are the expanded red regions.
However, the set modification might make the specification unrealizable, e.g., in Fig. \ref{pic:notraj} the robot cannot move from $\pi_4$ to $\pi_1$ while avoiding $\pi_2$, even though the specification can be realized on the workspace of the robot that the user perceives.
In this case, the user is entirely left in the dark as of why the specification failed and, more importantly, on what actually the system can achieve under these new constraints. 
This is especially important since the low level controller synthesis details should be hidden from the end user.
\exmend
\end{exmp}

\begin{figure}[t]
\centering
\includegraphics[width=6.5cm]{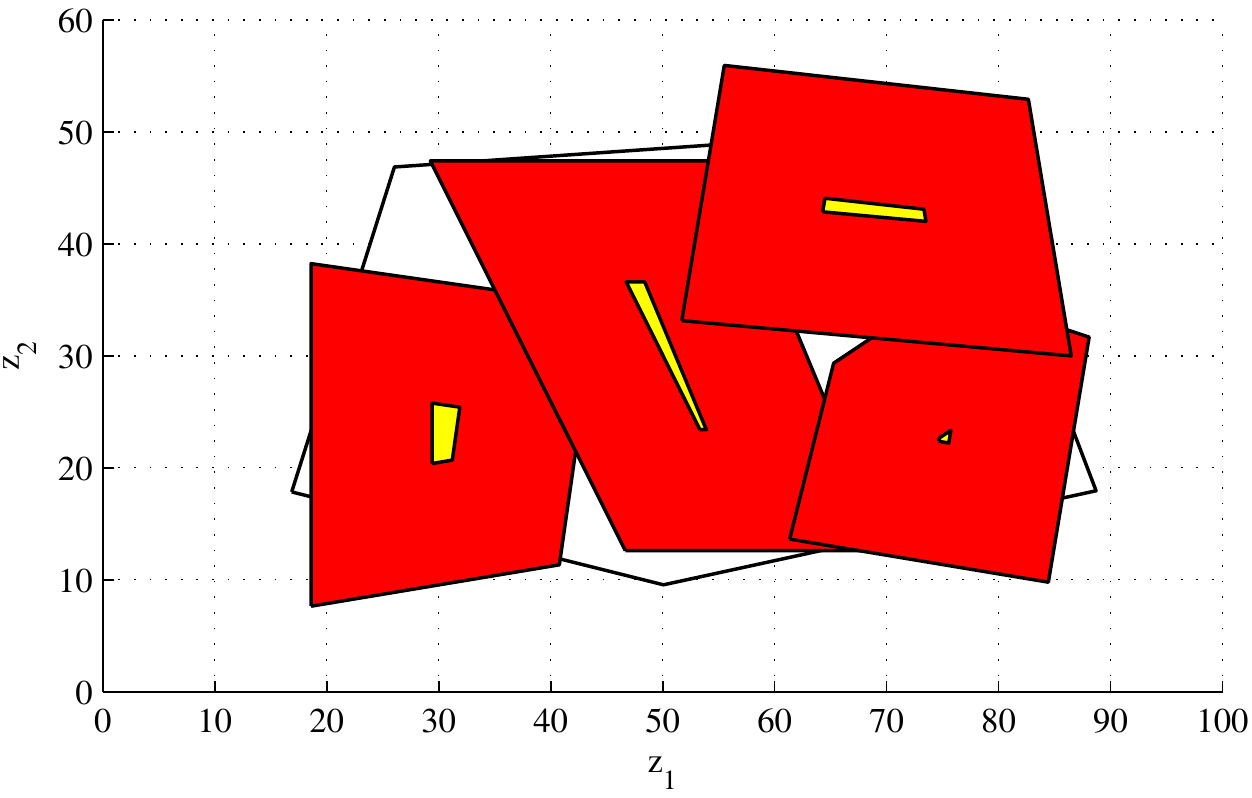}
\caption{The modified environment of Fig. \ref{pic:areas4} under large bounds on the permissible acceleration $U$.
The red regions indicate areas that should be avoided in order to satisfy $\neg \pi_i$ while the yellow regions indicate areas that should be visited in order to satisfy $\pi_i$.}
\label{pic:notraj}
\end{figure}

The next example presents a typical scenario for task planning with two agents.

\ifthenelse {\boolean{BGRAPHPDF}}
{
\begin{figure}[t]
\centering
\includegraphics[width=5.5cm]{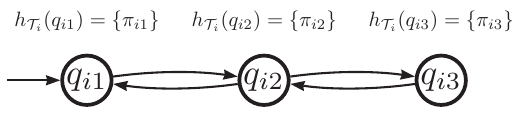}
\caption{Simple FSM model $\FTS_i$ of an autonomous agent. Each state $j$ of the FSM $i$ is labeled by an atomic proposition $\pi_{ij}$.} 
\label{fig:simp:sys}
\end{figure}
}
{
\begin{figure}
\begin{center}
\VCDraw{%
\begin{VCPicture}{(-1,-0.2)(7,1.8)}
\State[q_{i1}]{(0,0)}{A} \State[q_{i2}]{(3,0)}{B} \State[q_{i3}]{(6,0)}{C}
\Initial{A}  
\rput(3,1){$h_{\FTS_i}(q_{i2})=\{\pi_{i2}\}$}
\rput(0,1){$h_{\FTS_i}(q_{i1})=\{\pi_{i1}\}$}
\rput(6,1){$h_{\FTS_i}(q_{i3})=\{\pi_{i3}\}$}
\ncarc{->}{A}{B} 
\ncarc{->}{B}{A} 
\ncarc{->}{B}{C} 
\ncarc{->}{C}{B} 
\end{VCPicture}} 
\end{center}
\caption{Simple FSM model $\FTS_i$ of an object or an autonomous agent. Each state $q_{ij}$ of the FSM $i$ is labeled by an atomic proposition $\pi_{ij}$.}
\label{fig:simp:sys}
\end{figure}

}

\ifthenelse {\boolean{BGRAPHPDF}}
{
\begin{figure}[t]
\centering
\includegraphics[width=7cm]{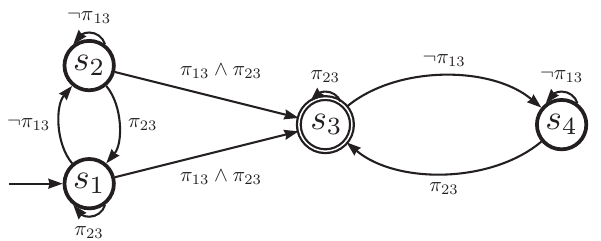}
\caption{The specification automaton of Example \ref{exm:twoagents}.}
\label{fig:exmp:init:spec2}
\end{figure}
}
{
\begin{figure}
\begin{center}
\VCDraw{%
\begin{VCPicture}{(0,-1.7)(8,1.5)}
\State[s_1]{(0,-1)}{A} \State[s_2]{(0,1)}{B} \FinalState[s_3]{(4,0)}{C}
\State[s_4]{(8,0)}{D}
\Initial{A}  
\ncline{->}{A}{C} \nbput{${\pi}_{13} \wedge \pi_{23}$}
\ncline{->}{B}{C} \naput{${\pi}_{13} \wedge \pi_{23}$}
\ncarc[arcangle=40]{->}{A}{B} \naput{$\neg \pi_{13}$}
\ncarc[arcangle=40]{->}{B}{A} \naput{$\pi_{23}$}
\ncarc[arcangle=40]{->}{C}{D} \naput{$\neg \pi_{13}$}
\ncarc[arcangle=40]{->}{D}{C} \naput{$\pi_{23}$}
\nccircle{<-}{A}{-0.28cm}\naput{${\pi}_{23}$}
\nccircle{->}{B}{0.28cm}\nbput{${\neg \pi}_{13}$}
\nccircle{->}{C}{0.28cm}\nbput{${\pi}_{23}$}
\nccircle{->}{D}{0.28cm}\nbput{$\neg \pi_{13}$}
\end{VCPicture}}
\end{center}
\caption{The specification automaton of Example \ref{exm:twoagents}.}
\label{fig:exmp:init:spec2}
\end{figure}
}

\begin{exmp}[Multi-Agent Planning] 
\label{exm:twoagents} 
We consider two autonomous agents whose independent actions can be modeled using an FSM as in Fig. \ref{fig:simp:sys}.
In this example, each state $q_i$ represents a location $i$.
In order to construct a simple example, we assume that at each discrete time instance only one agent can move.
Alternatively, we can think of these agents as being objects moved around by a mobile manipulator and that the manipulator can move only one object at a time.
This means that the state of both objects can be described by the asynchronous composition of the two state machines $\FTS_1$ and $\FTS_2$.
The asynchronous composition results in a FSM $\FTS$ with 9 states where each state $(q_{1i},q_{2j})$ is labeled by $h_{\FTS}(q_{1i},q_{2j}) = h_{\FTS_1}(q_{1i}) \cup h_{\FTS_2}(q_{2j})$.

The system must accomplish the task: ``Object 2 should be placed in Location 3 after Object 1 is placed in Location 3".
Note that this requirement could be used to enforce that Object 2 is going to be positioned on top of Object 1 at the end of the system execution.
However, the requirement permits temporary placing Object 2 in Location 3 before Object 1 is placed in Location 3.
This should be allowed for problems where a temporary reposition of the objects is necessary.
Now, let's assume that the aforementioned task is just a part from a long list or requirements which also include the task: ``Always, Object 1 should not be in Location 3 until Object 2 moves in Location 3".

These are informal requirements and in order to give them mathematical meaning we will have to use a formal language.
In Linear Temporal Logic (LTL) (see \citealt{ClarkeGP99}), the requirements become%
\footnote{Here, F stands for eventually in the future, $\Un$ for until and G for always. 
Further introduction of LTL is out of the scope of this paper and the interested reader can explore the logic in \citealt{FainekosGKGP09automatica}. 
We use LTL in the following for succinctness in the presentation.}:
F($\pi_{13}$ $\wedge$ F$\pi_{23}$) and G(($\neg$$\pi_{13}$)$\Un$$\pi_{23}$).
We remark that the conjunction of these two requirements is actually a satisfiable specification (even though the requirements appear conflicting) and the corresponding specification automaton is presented in Fig. \ref{fig:exmp:init:spec2}.
The specification remains satisfiable because the semantics of the logic permit both objects to be place in Location 3 at the same time (see transitions on label $\pi_{13} \wedge \pi_{23}$ from states $s_1$ and $s_2$ to state $s_3$).

However, there is no trajectory of the FSM $\FTS$ that will satisfy the specification. 
Recall that the model does not allow for simultaneous transitions of the two objects.
Again, the user does not know why the specification failed and, more importantly, on what actually the system can achieve that is ``close" to the initial user intent.
\exmend
\end{exmp}

When a specification $\Bc$ is not satisfiable on a particular system $\FTS$, then the current motion planning  and control synthesis methods (e.g., \citealt{FainekosGKGP09automatica,KloetzerB10tro,LaViersEtAl2011iccps}) based on automata theoretic concepts (see \citealt{GiacomoV99ecp}) simply return that the specification is not satisfiable without any other user feedback. 
In such cases, we would like to be able to solve the following problem and provide feedback to the user.

\begin{prob}[Minimal Revision Problem (MRP)]
Given a system $\FTS$ and a specification automaton $\Bc$, if the specification $\Bc$ cannot be satisfied on $\FTS$, then find the ``closest" specification $\Bc'$ to $\Bc$ which can be satisfied on $\FTS$.
\label{prb:main}
\end{prob}

Problem \ref{prb:main} was first introduced in \citealt{Fainekos11icra} for Linear Temporal Logic (LTL) specifications. 
In \citealt{Fainekos11icra}, we provided solutions to the debugging and (not minimal) revision problems and we demonstrated that we can easily get a minimal revision of the specification when the discrete controller synthesis phase fails due to unreachable states in the system.

\begin{ass}
All the states on $\FTS$ are reachable.
\label{ass:reach}
\end{ass}

In \citealt{KimFS12icra}, we introduced a notion of distance on a restricted space of specification automata and, then, we were able to demonstrate that MRP is in NP-complete even on that restricted search space of possible solutions.
Since brute force search is prohibitive for any reasonably sized problem, we presented an encoding of MRP as a satisfiability problem.
Nevertheless, even when utilizing state of the art satisfiability solvers, the size of the systems that we could handle remained  small (single robot scenarios in medium complexity environments).

In \citealt{KimFS12iros}, we provided an approximation algorithm for MRP. 
The algorithm is based on Dijkstra's single-source shortest path algorithm (see \citealt{CormenLRS01}), which can be regarded both as a greedy and a dynamic programming algorithm (see \citealt{Sniedovich06cc}).
We demonstrated through numerical experiments that not only the algorithm returns an optimal solution in various scenarios, but also that it outperforms in computation time our satisfiability based solution.
Then, we presented some scenarios where the algorithm is guaranteed not to return the optimal solution.


{\bf Contributions:}
In this paper, we define the MRP problem and we provide the proof that MRP is NP-complete even when restricting the search space (e.g., Problem \ref{prb:formal}). 
Then, we provide an approximation algorithm for MRP and theoretically establish the upper bound of the algorithm for a special case. 
Furthermore, we show that for our heuristic algorithm a constant approximation ratio cannot be established, in general.
We also present experimental results of the scalability of our framework and establish some experimental approximation bounds on random problem instances.
Finally, in order to improve the paper presentation, we also provide multiple examples that have not been published before.

%% file: planning.tex
\section{Preliminaries}
\label{sec:prelim}

In this section, we review some basic results on the automata theoretic planning and the specification revision problem from \citealt{FainekosGKGP09automatica,Fainekos11icra}.

\yhl{ %
Throughout the paper, we will use the notation $\mathcal{P}(A)$ for representing the powerset of a set $A$, i.e., $\mathcal{P}(A) = \{B \;|\; B \subseteq A\}$. 
Clearly, it includes $\emptyset$ and $A$ itself. We also define the set difference as $A \setminus B = \{x \in A\;|\;x \notin B\}$.
}

\subsection{Constructing Discrete Controllers}
\label{sec:ltlplanning}

We assume that the combined actions of the robot/team of robots and their operating environment can be represented using an FSM.
 
\begin{defn} [FSM] A Finite State Machine is a tuple $\FTS = (Q, Q_0, \rightarrow_\FTS, h_\FTS, \Pi)$ where:
\ifthenelse {\boolean{TECHREP}}
{
\begin{itemize}
\item 
}{}
$Q$ is a set of states; 
\ifthenelse {\boolean{TECHREP}}
{
\item 
}{}
$Q_0 \subseteq Q$ is the set of possible initial states; 
\ifthenelse {\boolean{TECHREP}}
{
\item 
}{}
$\rightarrow_\FTS \subseteq Q \times Q$ is the transition relation; and, 
\ifthenelse {\boolean{TECHREP}}
{
\item 
}{}
$h_\FTS : Q \rightarrow \Pc(\Pi)$ maps each state $q$ to the set of atomic propositions that are true on $q$.
\ifthenelse {\boolean{TECHREP}}
{
\end{itemize}
}{}
\end{defn}

We define a {\it path} on the FSM to be a sequence of states and a {\it trace} to be the corresponding sequence of sets of propositions.
Formally, a path is a function $p : \Ne \rightarrow Q$ such that for each $i\in \Ne$ we have $p(i) \rightarrow_\FTS p(i+1)$ and the corresponding trace is the function composition $\bar p = h_\FTS \circ p : \Ne \rightarrow \Pc(\Pi)$.
The language $\Lc(\FTS)$ of $\FTS$ consists of all possible traces.

\begin{exmp}
For the two agent system in Example \ref{exm:twoagents}, a path would be $(q_{11},q_{21})(q_{11},q_{22})(q_{12},q_{22}) \ldots$ and the corresponding trace would be $\{\pi_{11},\pi_{21}\}\{\pi_{11},\pi_{22}\}\{\pi_{12},\pi_{22}\} \ldots$.
\end{exmp}


In this work, we are interested in the $\omega$-automata that will impose certain requirements on the traces of $\FTS$.
Omega automata differ from the classic finite automata in that they accept infinite strings (traces of $\FTS$ in our case).

\begin{defn}
\label{def:buchi}
A automaton is a tuple $\BUCHI = (S_\BUCHI, s_{0}^{\BUCHI}, \Omega,  \yhl{\rightarrow_{\BUCHI}}, F_{\BUCHI})$ where:
\ifthenelse {\boolean{TECHREP}}
{
\begin{itemize}
\item 
}{}
$S_\BUCHI$ is a finite set of states;
\ifthenelse {\boolean{TECHREP}}
{
\item 
}{}
$s_0^\BUCHI$ is the initial state; 
\ifthenelse {\boolean{TECHREP}}
{
\item 
}{}
 $\Omega$ is an input alphabet;
\ifthenelse {\boolean{TECHREP}}
{
\item 
}{}
\yhl{ $\rightarrow_{\BUCHI} \subseteq S_\BUCHI \times \Omega \times S_\BUCHI$ is a transition relation;} and
\ifthenelse {\boolean{TECHREP}}
{
\item 
}{}
$F_{\BUCHI} \subseteq S_\BUCHI$ is a set of final states.
\ifthenelse {\boolean{TECHREP}}
{
\end{itemize}
}{}
\end{defn}

We also write $s \stackrel{l}{\rightarrow}_\BUCHI s'$ instead of $(s, l, s') \in \rightarrow_\BUCHI$.
A {\it specification} automaton is an automaton with B{\"u}chi acceptance condition where the input alphabet is the powerset of the labels of the system $\Tc$, i.e., $\Omega = \Pc(\Pi)$.

A {\it run} $r$ of a specification automaton $\BUCHI$ is a sequence of states $r : \Ne \rightarrow S_\BUCHI$ that occurs under an input trace $\bar p$ taking values in $\Omega = \Pc(\Pi)$.
That is, for $i = 0$ we have $r(0) = s_{0}^{\BUCHI}$ and for all $i \geq 0$ we have $r(i) \stackrel{\bar p(i)}{\rightarrow}_\BUCHI r(i+1)$.  
Let $\lim(\cdot)$ be the function that returns the set of states that are encountered infinitely often in the run $r$ of $\BUCHI$. 
Then, a run $r$ of an automaton $\BUCHI$ over an infinite trace $\bar p$ is {\it accepting} if and only if $\lim(r) \cap F_{\BUCHI} \neq \emptyset$.
This is called a B{\"u}chi acceptance condition.
Finally, we define the language $\Lc(\BUCHI)$ of $\BUCHI$ to be the set of all traces $\bar p$ that have a run that is accepted by $\BUCHI$.

\yhl{Even though the definition of specification automata (Def. \ref{def:buchi}) only uses sets of atomic propositions for labeling transitions, it is convenient for the user to read and write specification automata with propositional formulas on the transitions.
The popular translation tools from LTL to automata, e.g., \cite{GastinO01cav}, label the transitions with propositional formulas in Disjunctive Normal Form (DNF). 
A DNF formula on a transition of a specification automaton can represent multiple transitions between two states.
}
In the subsequent sections, we will be making the following simplifying assumption on the structure of the specification automata that we consider.

\begin{ass}
All the propositional formulas that appear on the transitions of a specification automaton are in Disjunctive Normal Form (DNF).
\yhl{That is, for any two states $s_1$, $s_2$ of an automaton $\BUCHI$, we represent the propositional formula that labels the corresponding transition by $\Phi_{\BUCHI}(s_1,s_2) = \bigvee_{i \in D_{s_1s_2}} \bigwedge_{j \in C^i_{s_1s_2}} \psi_{ij}$ for some appropriate set of indices $D_{s_1s_2}$ and $C^i_{s_1s_2}$.
Here, $\psi_{ij}$ is a literal which is $\pi$ or $\neg\pi$ for some $\pi \in \Pi$.
Finally, we assume that when any subformula in $\Phi_{\BUCHI}(s_1,s_2)$ is a tautology or a contradiction, then it is replaced by $\top$ (true) or $\bot$ (false), respectively.}
\label{ass:trans}
\end{ass}

\yhl{The last assumption is necessary in order to avoid converting a contradiciton like $\pi \wedge \neg \pi$ into a satisfiable formula (see Sec. \ref{sec:specrev}).}
We remark that the Assumption \ref{ass:trans} does not restrict the scope of this work.
Any propositional formula can be converted in DNF where any negation symbol appears in front of an atomic proposition.

\begin{exmp}
Let us consider the specification automaton in Fig. \ref{fig:exmp:init:spec2}.
The propositional formulas over the set of atomic propositions $\Pi$ are shorthands for the subsets of $\Pi$ that would label the corresponding transitions. 
For example, the label $\pi_{13}\wedge \pi_{23}$ over the edge $(s_2,s_3)$ succinctly represents all the transitions $(s_2,l,s_3)$ such that $\{\pi_{13},\pi_{23}\} \subseteq  l \subseteq \Pi$.
On the other hand, the label $\neg \pi_{13}$ over the edge $(s_1,s_2)$ succinctly represents all the transitions $(s_1,l,s_2)$ such that $l \subseteq \Pi$ and $\pi_{13}\not\in l$.
On input trace $\{\pi_{11},\pi_{21}\}\{\pi_{11},\pi_{22}\}\{\pi_{12},\pi_{22}\} \ldots$ the corresponding run would be $s_1 s_2 s_2 s_2 \ldots$. \exmend
\end{exmp}

In brief, our goal is to generate paths on $\FTS$ that satisfy the specification $\ASPEC$.
In automata theoretic terms, we want to find the subset of the language $\Lc(\FTS)$ which also belongs to the language $\Lc(\ASPEC)$.
This subset is simply the intersection of the two languages $\Lc(\FTS) \cap \Lc(\ASPEC)$ and it can be constructed by taking the product $\FTS \times \ASPEC$ of the FSM $\FTS$ and the specification automaton $\ASPEC$.
Informally, the automaton $\ASPEC$ restricts the behavior of the system $\FTS$ by permitting only certain acceptable transitions.
Then, given an initial state in the FSM $\FTS$, we can choose a particular trace from $\Lc(\FTS) \cap \Lc(\ASPEC)$ according to a preferred criterion.

\begin{defn}
The product automaton $\Ac = \FTS \times \ASPEC$ is the automaton $\Ac = (S_\Ac, s_{0}^{\Ac}, \Pc(\Pi), \yhl{\rightarrow_\Ac}, F_\Ac)$ where:
\begin{itemize}
\item $\Sc_\Ac = Q \times S_{\ASPEC}$,
\item $s_{0}^{\Ac} = \{(q_0,s_{0}^{\ASPEC}) \; | \; q_0 \in Q_0\}$,
\item \yhl{$\rightarrow_\Ac \subseteq S_\Ac \times \Pc(\Pi) \times S_\Ac$ s.t. $(q_i,s_i) \stackrel{l}{\rightarrow}_{\Ac} (q_j,s_j)$ iff  $q_i \rightarrow_\FTS q_j$ and  $s_i \stackrel{l}{\rightarrow}_{\ASPEC} s_j$ with $l = h_\FTS(q_j)$}, and
\item $F_\Ac = Q \times F_{\BUCHI}$ is the set of accepting states.
\end{itemize}
\end{defn}

Note that $\Lc(\Ac) = \Lc(\FTS) \cap \Lc(\ASPEC)$.
We say that $\ASPEC$ is {\it satisfiable} on $\FTS$ if $\Lc(\Ac) \neq \emptyset$.
Moreover, finding a satisfying path on $\FTS \times \ASPEC$ is an easy algorithmic problem (see \citealt{ClarkeGP99}). 
First, we convert automaton $\FTS \times \ASPEC$ to a directed graph and, then, we find the strongly connected components (SCC) in that graph. 

If at least one SCC that contains a final state is reachable from an initial state, then there exist accepting (infinite) runs on $\FTS \times \ASPEC$ that have a finite representation.
Each such run consists of two parts: {\bf prefix:} a part that is executed only once (from an initial state to a final state) and, {\bf lasso:} a part that is repeated infinitely (from a final state back to itself). 
Note that if no final state is reachable from the initial or if no final state is within an SCC, then the language $\Lc(\Ac)$ is empty and, hence, the high level synthesis problem does not have a solution.
{\it Namely, the synthesis phase has failed and we cannot find a system behavior that satisfies the specification.}

\begin{exmp}
The product automaton of Example \ref{exm:twoagents} has 36 states and 240 number of transitions. However, no final state is reachable from the initial state. \exmend
\end{exmp}

%% file: automaton_revision_sat.tex
\section{The Specification Revision Problem}
\label{sec:specrev}


Intuitively, a revised specification is one that can be satisfied on the discrete abstraction of the workspace or the configuration space of the robot.
In order to search for a minimal revision, we need first to define an ordering relation on automata as well as a distance function between automata.
Similar to the case of LTL formulas in \citealt{Fainekos11icra}, we do not want to consider the ``space" of all possible automata, but rather the ``space" of specification automata which are semantically close to the initial specification automaton $\ASPEC$.
The later will imply that we remain close to the initial intention of the designer.
We propose that this space consists of all the automata that can be derived from $\ASPEC$ by relaxing \yhl{the restrictions for transitioning from one state to another}.
\yhl{In other words, we introduce possible transitions between two states of the specification automaton.}

\begin{exmp}
\yhl{
Consider the specification automaton $\ASPEC$ and the automaton $\BUCHI_1$ in Fig. \ref{fig:valid_rel_B}.
The transition relations of the two automata are defined as:
}
\begin{align*}
\text{\yhl{$ \rightarrow_{\ASPEC} = $}} 
& 
\text{\yhl{
$ \stackrel{0\xrightarrow{\neg \pi_0}_{\ASPEC} 0 }{\overbrace{\{ (0, l, 0) \; | \; l \subseteq \Pi \text{ and } \pi_0 \not \in l \}}} 
\cup 
\stackrel{0\xrightarrow{\pi_2}_{\ASPEC} 0 }{\overbrace{\{(0, l ,0) \; | \; \{\pi_2\} \subseteq l \subseteq \Pi \}}} 
\cup 
\stackrel{0\xrightarrow{\neg \pi_0 \wedge \pi_1}_{\ASPEC} 1 }{\overbrace{\{(0, \{\pi_1\}, 1), (0, \{\pi_1, \pi_2\}, 1) \}}}
\cup
$
}}
\displaybreak[3] \\
& \text{\yhl{
$ 
 \cup 
 \stackrel{1\xrightarrow{\neg \pi_0 }_{\ASPEC} 1 }{\overbrace{\{ (1, l, 1) \; | \; l \subseteq \Pi \text{ and } \pi_0 \not \in l \}}}  
 \cup  
\stackrel{1\xrightarrow{\neg \pi_0 \wedge \pi_2}_{\ASPEC} 2 }{\overbrace{\{ (1, \{\pi_2 \} ,2), (1, \{\pi_1, \pi_2 \} ,2)  \}}}
 \cup 
 \stackrel{0\xrightarrow{\neg \pi_0 \wedge \pi_1 \wedge \pi_2}_{\ASPEC} 2 }{\overbrace{\{(0, \{\pi_1, \pi_2 \} ,2)  \}}}
 \cup$
 }} 
\displaybreak[3] \\
& \text{\yhl{
$ \cup 
\stackrel{2\xrightarrow{\neg \pi_0}_{\ASPEC} 2 }{\overbrace{\{ (2, l, 2) \; | \; l \subseteq \Pi \text{ and } \pi_0 \not \in l \}}} $
}} 
\displaybreak[3] \\
%
\text{\yhl{$ \rightarrow_{\BUCHI_1} = $}}   
&
 \text{\yhl{ 
 $ \stackrel{0\xrightarrow{\top}_{\BUCHI_1} 0 }{\overbrace{\{ (0, l, 0) \; | \; l\subseteq \Pi\} }}
  \cup 
  \stackrel{0\xrightarrow{\neg \pi_0 \wedge \pi_1}_{\BUCHI_1} 1 }{\overbrace{\{ (0, \{\pi_1\} ,1), (0, \{\pi_1, \pi_2\}, 1)\}}} 
  \cup 
  \stackrel{1\xrightarrow{\top}_{\BUCHI_1} 1 }{\overbrace{\{ (1, l, 1) \; | \; l\subseteq \Pi\} }} 
  \cup 
  \stackrel{1\xrightarrow{\neg \pi_0 \wedge \pi_2}_{\BUCHI_1} 2 }{\overbrace{\{ (1, \{\pi_2 \} ,2), (1, \{\pi_1, \pi_2 \} ,2)  \}}}
   \cup $
   }} 
   \\
& \text{\yhl{ 
$ \cup 
\stackrel{1\xrightarrow{\pi_2}_{\BUCHI_1} 2 }{\overbrace{\{ (1, l, 2) \; | \; \{\pi_2\}\subseteq l \subseteq \Pi\} }} 
\cup 
\stackrel{2\xrightarrow{\neg \pi_0}_{\BUCHI_1} 2 }{\overbrace{\{ (2, l, 2) \; | \;  l \subseteq \Pi \text{ and } \pi_0 \not\in l \}}}$
}}
\displaybreak[3] \\
\end{align*}
\yhl{
and, hence, $ \rightarrow_{\ASPEC} \subset \rightarrow_{\BUCHI_1} $.
In other words, any transition allowed by $\ASPEC$ is also allowed by $\BUCHI_1$ and, thus, any trace accepted by $\ASPEC$ is also accepted $\BUCHI_1$.
If we view this from the perspective of the specification, then this means that the specification imposes less restrictions or that the specification automaton $\BUCHI_1$ is a relaxed specification with respect to $\ASPEC$.
}
\label{exmp:relax:simple}
\end{exmp}

\yhl{
As the previous example indicated, specification relaxation could be defined as the subset relation between the transition relations of the specification automata.
However, this is not sufficient from the perspective of user requirements.
For instance, consider again Example \ref{exmp:relax:simple}.
The transition $\mycirc{1}\xrightarrow{\neg \pi_0 \wedge \pi_2}_{\ASPEC} \mycirc{2}$ could be relaxed as $\mycirc{1} \xrightarrow{(\neg \pi_0 \wedge \pi_2) \vee \pi_3}_{\ASPEC} \mycirc{2}$. 
A relevant relaxation from the user perspective should be removing either of the constraints $\neg \pi_0$ or $\pi_2$ rather than introducing a new requirement $\pi_3$. 
Introducing $\pi_3$ may implicitly relaxe both constraints $\neg \pi_0$ and $\pi_2$ in certain contexts.
However, our goal in this paper  is to find a minimal relaxation.
}

\yhl{In Sec. \ref{sec:prelim}, we indicated that a transition relation could be compactly represented using a proposition formula $\Phi_{\BUCHI}(s_1,s_2) = \bigvee_{i \in D_{s_1s_2}} \bigwedge_{j \in C^i_{s_1s_2}} \psi_{ij}$ in DNF. 
Given a pair of states $(s_1,s_2)$ and a set of indices $\hat D_{s_1s_2}$, $\hat C^i_{s_1s_2}$, we define a {\it substitution} $\theta$ as 
\[\theta(s_1,s_2,\psi_{ij}) = \left \{
\begin{array}{ll}
\top  & \text{ if } i \in \hat D_{s_1s_2} \text{ and } j \in \hat C^i_{s_1s_2}\\
\psi_{ij}  & \text{ otherwise } \\
\end{array}
\right.
\]
That is, a substitution only relaxes the constraints on the possible transitions between two automaton states.
}
 
\begin{defn}[Relaxation]
Let $\BUCHI_1 = (S_{\BUCHI_1}$, $s_0^{\BUCHI_1}$, $\Omega$, $\rightarrow_{\BUCHI_1}$, $F_{\BUCHI_1})$ and $\BUCHI_2 = (S_{\BUCHI_2}, s_{0}^{\BUCHI_2}, \Omega, \rightarrow_{\BUCHI_2}, F_{\BUCHI_2})$  be two B{\"u}chi automata. 
Then, we say that $\BUCHI_2$ is a relaxation of $\BUCHI_1$ and we write $\BUCHI_1 \preceq \BUCHI_2$ if and only if (1) $S_{\BUCHI_1} = S_{\BUCHI_2} = S$, (2) $s_{0}^{\BUCHI_1} = s_{0}^{\BUCHI_2}$, (3) $F_{\BUCHI_1} = F_{\BUCHI_2}$ and \yhl{ (4) there exists a substitution $\theta$ such that for all $(s,s') \in S\times S$, we have  $\Phi_{\BUCHI_2}(s,s') \equiv \bigvee_{i \in D_{ss'}} \bigwedge_{j \in C^i_{ss'}} \theta(s,s',\psi_{ij})$ when $\Phi_{\BUCHI_1}(s,s') = \bigvee_{i \in D_{ss'}} \bigwedge_{j \in C^i_{ss'}} \psi_{ij}$.}
\label{def:valid_rel_B}
\end{defn}

\ifthenelse {\boolean{BGRAPHPDF}}
{
\begin{figure}[t]
\centering
\includegraphics[width=8cm]{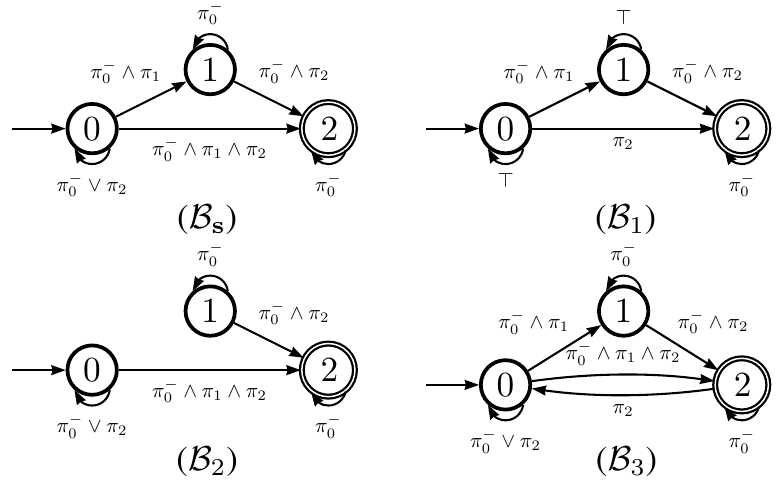}
\caption{Example \ref{exmp:valid_rel_B}. $\ASPEC$: the initial specification automaton,
here $\pi_0^- \equiv \neg \pi_0$; 
$\ASPEC \preceq \BUCHI_1$; 
$\ASPEC \parallel \BUCHI_2$; 
$\ASPEC \parallel \BUCHI_3$. 
}
\label{fig:valid_rel_B}
\end{figure}
}
{
\begin{figure}
\begin{center}
\VCDraw{%
\begin{VCPicture}{(-1,-1)(10,2.5)}
\State[0]{(-1,0)}{A} \State[1]{(1,1)}{B} \FinalState[2]{(3,0)}{C}
\Initial{A}  
\ncline{->}{A}{B} \naput{${\pi}_0^- \wedge \pi_1$}
\ncline{->}{A}{C} \nbput{${\pi}_0^- \wedge \pi_1 \wedge \pi_2$}
\ncline{->}{B}{C} \naput{${\pi}_0^- \wedge \pi_2$}
\nccircle{<-}{A}{-0.3cm}\naput{${\pi}_0^- \vee {\pi}_2$}
\nccircle{->}{B}{0.3cm}\nbput{${\pi}_0^-$}
\nccircle{<-}{C}{-0.3cm}\naput{${\pi}_0^-$}
\State[0]{(6,0)}{0} \State[1]{(8,1)}{1} \FinalState[2]{(10,0)}{2}  
\Initial{0}  
\ncline{->}{0}{1} \naput{${\pi}_0^- \wedge \pi_1$} 
\ncline{->}{1}{2} \naput{${\pi}_0^- \wedge \pi_2$} 
\ncline{->}{0}{2} \nbput{$\pi_2$}
\nccircle{<-}{0}{-0.3cm}\naput{$\top$}
\nccircle{->}{1}{0.3cm}\nbput{$\top$}
\nccircle{<-}{2}{-0.3cm}\naput{${\pi}_0^-$}
\end{VCPicture}}\\
\vspace{5pt}
($\ASPEC$) \qquad \qquad \qquad \qquad \qquad ($\BUCHI_1$) \\
\vspace{5pt}
\VCDraw{%
\begin{VCPicture}{(-1,-1)(10,2)}
\State[0]{(-1,0)}{A} \State[1]{(1,1)}{B} \FinalState[2]{(3,0)}{C}
\Initial{A}  
\ncline{->}{A}{C} \nbput{${\pi}_0^- \wedge \pi_1 \wedge \pi_2$}
\ncline{->}{B}{C} \naput{${\pi}_0^- \wedge \pi_2$}
\nccircle{<-}{A}{-0.3cm}\naput{${\pi}_0^- \vee {\pi}_2$}
\nccircle{->}{B}{0.3cm}\nbput{${\pi}_0^-$}
\nccircle{<-}{C}{-0.3cm}\naput{${\pi}_0^-$}
\State[0]{(6,-0.25)}{0} \State[1]{(8,1)}{1} \FinalState[2]{(10,-0.25)}{2}  
\Initial{0}  
\ncline{->}{0}{1} \naput{${\pi}_0^- \wedge \pi_1$} 
\ncline{->}{1}{2} \naput{${\pi}_0^- \wedge \pi_2$} 
\ncarc{->}{0}{2} \naput{$\pi_0^- \wedge \pi_1 \wedge \pi_2$}
\ncarc{->}{2}{0} \naput{$\pi_2$}
\nccircle{<-}{0}{-0.3cm}\naput{$\pi_0^- \vee {\pi}_2$}
\nccircle{->}{1}{0.3cm}\nbput{$\pi_0^-$}
\nccircle{<-}{2}{-0.3cm}\naput{${\pi}_0^-$}
\end{VCPicture}} \\
\vspace{5pt}
($\BUCHI_2$) \qquad \qquad \qquad \qquad \qquad ($\BUCHI_3$) \\
\end{center}
\caption{(Example \ref{exmp:valid_rel_B}) $\ASPEC$: the initial specification automaton,
here $\pi_0^- \equiv \neg \pi_0$; 
$\ASPEC \preceq \BUCHI_1$; 
$\ASPEC \parallel \BUCHI_2$; 
$\ASPEC \parallel \BUCHI_3$. 
}
\label{fig:valid_rel_B}
\end{figure}
}

\yhl{In Def. \ref{def:valid_rel_B}, when defining $\Phi_{\BUCHI_2}(s,s')$, we use the equivalence relation $\equiv$ rather then equality $=$ in order to highlight that in the resulting DNF formula no constant $\top$ appears in a subformula. 
The only case where $\top$ can appear is when $\Phi_{\BUCHI_2}(s,s') \equiv \top$.}
We remark that if $\BUCHI_1 \preceq \BUCHI_2$, then $\Lc(\BUCHI_1) \subseteq \Lc(\BUCHI_2)$ since the relaxed automaton allows more behaviors to occur%
\footnote{ \yhl{ Note that $\BUCHI_1 \preceq \BUCHI_2$ implies that $\BUCHI_2$ simulates $\BUCHI_1$ under the usual notion of simulation relation \cite{ClarkeGP99}.
However, clearly, if $\BUCHI_2$ simulates $\BUCHI_1$, then we cannot infer that $\BUCHI_2$ is a relaxation of $\BUCHI_1$ as defined in Def. \ref{def:valid_rel_B}.}}%
.
If two automata $\BUCHI_1$ and $\BUCHI_2$ cannot be compared under relation $\preceq$, then we write $\BUCHI_1 \parallel \BUCHI_2$.
The intuition behind the ordering relation in Def. \ref{def:valid_rel_B} is better explained by an example.

\begin{exmp}[Continuing from Example \ref{exmp:relax:simple}]
Consider the specification automaton $\ASPEC$ and the automata $\BUCHI_1$-$\BUCHI_3$ in Fig. \ref{fig:valid_rel_B}.
Definition \ref{def:valid_rel_B} specifies that the two automata must have transitions between exactly the same states%
\footnote{ To keep the presentation simple, we do not extend the definition of the ordering relation to isomorphic automata. 
Also, this is not required in our technical results since we are actually going to construct automata which are relaxations of a specification automaton.
The same holds for bisimilar automata (e.g., \citealt{Park1981cai}.)}%
.
\yhl{Moreover, if the propositional formula that labels a transition between the same pair of states on the two automata differs, then the propositional formula on the relaxed automaton must be derived by the corresponding label of the original automaton by removing literals. 
The latter means that we have relaxed the constraints that permit a transition on the specification automaton.}


\yhl{By visual inspection of $\ASPEC$ and $\BUCHI_1$ in Fig. \ref{fig:valid_rel_B}, we see that $\ASPEC \preceq \BUCHI_1$.
For example, the transition $\mycirc{0} \xrightarrow{\pi_2}_{\BUCHI_1} \mycirc{2}$ is derived from $\mycirc{0} \xrightarrow{\neg \pi_0 \wedge \pi_1 \wedge \pi_2}_{\ASPEC} \mycirc{2}$ by replacing the literals $\neg \pi_0$ and $\pi_1$ with $\top$.
Similarly, we have $\ASPEC \parallel \BUCHI_2$ since $\Phi_{\ASPEC}(0,1) = \neg \pi_0 \wedge \pi_1$ and $\Phi_{\BUCHI_2}(0,1) = \bot$, i.e., we have removed a transition between two states.
We also have $\ASPEC \parallel \BUCHI_3$ since $\Phi_{\ASPEC}(2,0) = \bot$ and $\Phi_{\BUCHI_3}(2,0) = \pi_2$, i.e., we have added a transition between two states.}
\exmend
\label{exmp:valid_rel_B}
\end{exmp}

\yhl{We remark that we restrict the space of relaxed specification automata to all automata that have the same number of states, the same initial state, and the same set of final states for computational reasons.
Namely, under these restrictions, we can convert the specification revision problem into a graph search problem.
Otherwise, the graph would have to me mutated by adding and/or removing states.}

We can now define the set of automata over which we will search for a minimal solution that has nonempty intersection with the system.

\begin{defn}
Given a system $\FTS$ and a specification automaton $\ASPEC$, the set of {\it valid relaxations} of $\ASPEC$ is defined as 
$\Rr(\ASPEC,\FTS) = \{ \BUCHI \; | \; \ASPEC \preceq \BUCHI \mbox{ and } \Lc(\FTS \times \BUCHI) \neq \emptyset \}.$
\end{defn}

We can now search for a minimal solution in the set $\Rr(\ASPEC,\FTS)$.
That is, we can search for some $\BUCHI \in \Rr(\ASPEC,\FTS)$ such that if for any other $\BUCHI' \in \Rr(\ASPEC,\FTS)$, we have $\BUCHI' \preceq \BUCHI$, then $\Lc(\BUCHI) = \Lc(\BUCHI')$.
However, this does not imply that a minimal solution semantically is minimal structurally as well.
In other words, it could be the case that $\BUCHI_1$ and $\BUCHI_2$ are minimal relaxations of some $\ASPEC$, but $\BUCHI_1 \parallel \BUCHI_2$ and, moreover, $\BUCHI_1$ requires the modification of only one transition while $\BUCHI_2$ requires the modification of two transitions. 
Therefore, we must define a metric on the set $\Rr(\ASPEC,\FTS)$, which accounts for the number of changes from the initial specification automaton $\ASPEC$. 


\yhl{
\begin{defn} (Distance)
Given a system $\FTS$ and a specification automaton $\ASPEC$, we define the distance of any $\BUCHI \in \Rr(\ASPEC,\FTS)$ that results from $\ASPEC$ under substitution $\theta$ to be $\mathbf{dist}_{\ASPEC}(\BUCHI)$ = $\sum_{(s,s')\in E_{\ASPEC}} \sum_{i \in \hat D_{ss'}} | \hat C^i_{ss'} |$ where $| \cdot |$ is the cardinality of the set.
\label{defn:dist}
\end{defn}
}

\yhl{We remark that given two relaxations $\BUCHI_1$ and $\BUCHI_2$ of some $\ASPEC$ where $\BUCHI_1 \preceq \BUCHI_2$, but $\BUCHI_2 \npreceq \BUCHI_1$, then $\mathbf{dist}_{\ASPEC}(\BUCHI_1) \leq \mathbf{dist}_{\ASPEC}(\BUCHI_2)$.
}

Therefore, Problem \ref{prb:main} can be restated as:

\begin{prob}
Given a system $\FTS$ and a specification automaton $\ASPEC$ such that $\Lc(\FTS \times \ASPEC) = \emptyset$, find $\BUCHI \in \arg\min\{ \mathbf{dist}_{\ASPEC}(\BUCHI') \; | \; \BUCHI' \in \Rr(\ASPEC,\FTS) \}$.
\label{prb:formal}
\end{prob}

\subsection{Minimal Revision as a Graph Search Problem}

In order to solve Problem~\ref{prb:formal}, we construct a directed labeled graph $G_\Ac$ from the product automaton $\Ac = \FTS \times \ASPEC$.
The edges of $G_\Ac$ are labeled by a set of atomic propositions which if removed from the corresponding transition on $\ASPEC$, they will enable the transition on $\Ac$.
The overall problem then becomes one of finding the least number of atomic propositions to be removed in order for the product graph to have an accepting run. 
Next, we provide the formal definition of the graph $G_{\Ac}$ which
corresponds to a product automaton $\Ac$ while considering the effect
of revisions.

\yhl{
To formally define the graph search problem, we will need some additional notation.
We first create two new sets of symbols from the set of atomic propositions $\Pi$:
\begin{itemize}
\item $\Pi^- = \{\pi^-\;|\; \pi \in \Pi\}$.
\item $\widetilde{\Pi} = \{l^+ \cup l^-\;|\;l^+ \subseteq \Pi \text{ and } l^- \subseteq \xi_S(\Pi \setminus l^+)\}$.
\end{itemize}
Given a transition between two states $s_1$ and $s_2$ of some $\BUCHI$ and a formula in DNF on the transition, we denote: 
\begin{itemize}
\item $\xi_l(\psi)$ = $\pi$ if $\psi = \pi$, and $\pi^-$ if $\psi = \neg\pi$.
\item $\xi_S(l) = \{\pi^- \;| \; \pi \in l\}$ where $l \subseteq \Pi$.
\end{itemize}
Now, we can introduce the following notation. 
We define
\begin{itemize}
\item the set $E_\BUCHI \subseteq S_\BUCHI^2$, such that $(s,s') \in E_\BUCHI$ iff $\exists l \in \mathcal{P}(\Pi)$ , $s \stackrel{l}{\rightarrow}_\BUCHI s'$; and, 
\item the function $\lambda_{\BUCHI}(s_1,s_2) = \{\{\xi_l(\psi_{ij})\;|\;j \in C^i_{s_1s_2}\}\;|\; i \in D_{s_1s_2} \}$ as a transition function that maps a pair of states to the set of symbols that represent conjunctive clause.
\end{itemize}
}

\yhl{
That is, if $(s,s') \in E_\BUCHI$, then $\exists l \in \Pc(\Pi)$. $s \stackrel{l}{\rightarrow}_\BUCHI s'$; and if $(s,s') \not \in E_{\BUCHI}$, then $\lambda_{\BUCHI}(s,s') = \emptyset$.
Also, if $\lambda_{\BUCHI}(s,s') \neq \emptyset$, then $\forall l \in \lambda_{\BUCHI}(s,s') $ . $l \subseteq \Pc(\widetilde{\Pi})$. 
}

\begin{exmp}
\label{exmp:new:notation:1}
\yhl{
Consider a set $\Pi = \{\pi_0, \pi_1\}$. Then, $\widetilde{\Pi} = \{\emptyset, \{\pi_0\}, \{\pi^-_0\}, \{\pi_1\}, \{\pi^-_1\}, \{\pi_0, \pi^-_1\}, \{\pi^-_0,\pi_1\}, \Pi, \Pi^-\}$.
Given a transition $(s_1,s_2)$ of $\BUCHI$ and same $\Pi$, consider $\Phi_\BUCHI(s_1,s_2) = \neg\pi_0 \vee \pi_1$ on that transition. Then, $\lambda_{\BUCHI}(s_1,s_2) = \{\{\pi^-_0\},\{\pi_1\}\}$. If $\Phi_\BUCHI(s_1,s_2) = \neg\pi_0 \wedge \pi_1$, then $\lambda_{\BUCHI}(s_1,s_2) = \{\{\pi^-_0,\pi_1\}\}$.
}
\end{exmp}

\begin{defn}
  Given a system $\FTS$ and a specification automaton $\ASPEC$, we
  define the graph $G_{\Ac} = (V,E,v_s,V_f,\APRem,\Lambda_S)$, which corresponds to
  the product $\Ac = \FTS \times \ASPEC$ as follows
\begin{itemize}
\item $V = \Sc$ is the set of nodes;
\item $E = E_\Ac \cup E_D \subseteq \Sc \times \Sc$, where $E_\Ac$ is
  the set of edges that correspond to transitions on $\Ac$, i.e.,
  $((q,s),(q',s')) \in E_\Ac$ iff $\exists l \in \Pc(\Pi)$ . $(q,s)
  \stackrel{l}{\rightarrow}_\Ac (q',s')$; and $E_D$ is the set of
  edges that correspond to disabled transitions, i.e.,
  $((q,s),(q',s')) \in E_D$ iff $q \rightarrow_\FTS q'$ and \yhl{$(s,s') \in E_{\ASPEC}$, but there does not exist $(s,l,s') \in \rightarrow_{\ASPEC}$ such that $l = h_\FTS(q'))$};
\item $v_s = s_0^{\Ac}$ is the source node;
\item $V_f = F_\Ac$ is the set of sinks;
\item $\APRem = \{ \tupleof{\pi,(s,s')} \; | \; \pi \in \text{ \yhl{ $\widetilde{\Pi}$ } }, (s,s') \in E_{\ASPEC} \}$
\item \yhl{$\Lambda_S : E \rightarrow \Pc(\APRem)$} is the edge labeling function such that if $e = ((q,s),
(q',s'))$, then
\[
\text{\yhl{$\Lambda_S(e) = \{ \tupleof{l',(s,s')} \; | \;  l \in \lambda_{\ASPEC}(s, s') \text{, } l^+ = (l \cap \Pi) \setminus h_\FTS(q') \text{, } l^- = (l \cap \Pi^-) \setminus \xi_S(\Pi \setminus h_\FTS(q')) \text{ and } l' = l^+ \cup l^-\}$.}}
\]
\end{itemize}
\label{def:graph}
\end{defn}

\begin{exmp}[Continuing Example \ref{exmp:new:notation:1}]
\yhl{
We will derive $\Lambda_S(e)$ for an edge $e=((q,s_1),(q',s_2))$. 
Assume $\Phi_{\BUCHI}(s_1,s_2) = \neg\pi_0 \vee \pi_1$, then $\lambda_{\BUCHI}(s_1,s_2) = \{\{\pi^-_0\},\{\pi_1\}\}$. 
If $h_\FTS(q') = \{\pi_0\}$, then for $l = \{\pi^-_0\}$, $l^+ = \emptyset$, $l^- = \{\pi^-_0\}$ and $l' = \{\pi^-_0\}$, and for $l = \{\pi_1\}$, $l^+ = \{\pi_1\}$, $l^- = \emptyset$ and $l' = \{\pi_1\}$.
Thus, $\Lambda_S(e) = \{\tupleof{\{\pi^-_0\}, (s_1,s_2)}, \tupleof{\{\pi_1\},(s_1,s_2)}\}$.
}

\yhl{
Now assume $\Phi(s_1,s_2) = \neg\pi_0 \wedge \pi_1$, then $\lambda_{\BUCHI}(s_1,s_2) = \{\{\pi^-_0,\pi_1\}\}$.
If $h_\FTS(q') = \{\pi_2\}$, then $l = \{\pi^-_0, \pi_1\}$, $l^+ = \{\pi_1\} \setminus \{\pi_2\} = \{\pi_1\}$, $l^- = \{\pi^-_0\} \setminus \{\pi^-_0, \pi^-_1\} = \emptyset$, and $l' = \{\pi_1\}$.
Thus, $\Lambda_S(e) = \{\tupleof{\{\pi_1\}, (s_1,s_2)} \}$.
} 
\end{exmp}

\yhl{
In order to determine which atomic propositions we must remove from a transition of the specification automaton, we need to make sure that we can uniquely identify them.
Recall that $\Lambda_S$ returns a set, e.g., $\Lambda_S(e) = \{\tupleof{\{\pi^-_0,\pi_1\}, (s,s')}, \tupleof{\{\pi_1\},(s,s')}\}$.
Saying that we need to remove $\pi_1$ from the label of $(s,s')$, it may not be clear which element of the set $\Lambda_S$ $\pi_1$ refers to.
This affects both the theoretical connection of the graph $G_{\Ac}$ as a tool for solving Problem \ref{prb:formal} and the practical implementation of any graph search algorithm (e.g., see line \ref{alg:aamrp:sub:relax} of Alg. \ref{alg:aamrp:sub} in Sec. \ref{sec:heuristic}).
}

\yhl{
Thus, in the following, we assume that $G_{\Ac}$ uses a function $\Lambda : E \rightarrow \APRem$ instead of $\Lambda_S$.
Now, $\Lambda$ maps each edge $e$ to just a single tuple $\tupleof{l,(s,s')}$ instead of a set as in Def. \ref{def:graph}.
This can be easily achieved by adding some dummy states in the graph with incoming edges labeled by the tuples in the original set $\Lambda_S(e)$. 
We can easily convert the original graph $G_\Ac$ to the modified one. 
First, for each edge $e = ((q,s),(q',s')) \in E_D$, we add $|\Lambda_S((q,s),(q',s'))|$ new nodes $\tilde{v}^e_i$, $i = 1, \ldots, |\Lambda_S((q,s),(q',s'))|$. 
Second, we add the edge $((q,s),\tilde{v}^e_i)$ to $E$ and set each label $\Lambda((q,s),\tilde{v}^e_i)$ with exactly one $\tupleof{l,(s,s')} \in \Lambda_S((q,s),(q',s'))$.
Then, we add the edges $(\tilde{v}^e_i,(q',s'))$ to $E$ and set $\Lambda(\tilde{v}^e_i,(q',s'))$ $=$ $\tupleof{\emptyset,(s,s')}$. 
Finally, we repeat until all the edges are labeled by tuples rather than sets.
}

\yhl{
We remark that every time we add a new node $v^e_i$ for an edge that corresponds to the same transition $(s,s')$ of the specification automaton, then we use the same index $i$ for each member of $\Lambda_S(e)$.
This is so that later we can map each edge of the modified graph $G_{\Ac}$ to the correct clause in the DNF formula of the specification automaton.
The total number of new nodes that we need to add depends on the number of disjunctions on each label of the specification automaton and the structure of the FSM.
}

Now, if for some edge $e$, $\Lambda(e) \not= \emptyset$, then $\Lambda(e)$ specifies those atomic propositions in $\lambda_{\ASPEC}(s,s')$ that need to be removed in order to enable the edge in the product state of $\Ac$.
Note that the labels of the edges of $G_\Ac$ are elements of $\APRem$ rather than subsets of $\Pi$.
This is due to the fact that we are looking into removing an atomic proposition $\pi$ from a specific transition $(s,l,s')$ of $\ASPEC$ rather than all occurrences of $\pi$ in $\ASPEC$.



\yhl{
In the following, we assume that $v_s \not \in V_f$, otherwise, we would not have to revise the specification.
Furthermore, we define $|\tupleof{l,(s,s'),i}| = |l|$, i.e., the size of a tuple $\tupleof{l,(s,s'),i} \in \widetilde{\Pi} \times E_{\ASPEC} \times \Ne$ is defined to be the size of the set $l$.
}

\begin{defn}(Path Cost)
\yhl{
For some $n>0$, let $\mu = e_1 e_2 \ldots e_n$ be a finite path on the graph $G_{\Ac}$ that consists of edges of $G_{\Ac}$.
Let 
\[ \widetilde{\Lambda}(\mu) = \left \{ \tupleof{l,(s,s'),k} \; |\; l = \cup_{j \in J} l_j, J \subseteq \{ 1\ldots n\},  \forall j \in J, \Lambda(e_j) = \tupleof{l_j,(s,s')}, k =  
\left \{
\begin{array}{ll}
i & \text{if } e_j = ((q,s),\tilde{v}^{((q,s),(q',s'))}_i) \\
0 & \text{otherwise}
\end{array}
\right .
\right \}.\]
We define the cost of the path $\mu$ to be 
\[ Cost(\mu) = \sum_{\lambda \in \widetilde{\Lambda}(\mu)} |\lambda|. \]
}
\end{defn}

\yhl{
In the above definition, $\widetilde{\Lambda}(\mu)$ collects in the same tuple all the atomic propositions that must be relaxed in the same transition of the specification automaton.
It is easy to see that given some set $\widetilde{\Lambda}(\mu)$, then we can construct a substitution $\theta_\mu$ for the corresponding relaxed specification automaton since all the required information is contained in the members of $\widetilde{\Lambda}(\mu)$.
}

\begin{exmp}
Consider the Example in Fig. \ref{fig:product}.  
In the figure, we provide a partial description of an FSM $\FTS$, a specification automaton $\ASPEC$ and the corresponding product automaton $\Ac$.
The dashed edges indicate disabled edges which are labeled by the  atomic propositions that must be removed from  the specification in order to enable the transition on the system.
\yhl{
In this example, we do not have to add any new nodes since we have only one conjunctive clause on the transition of the specification automaton.
}  

\yhl{
It is easy to see now that in order to enable the path $(q_0,s_1),(q_1,s_1),(q_3,s_1)$ on the product automaton, we need to replace $\pi_0$ and $\pi_2$ with $\top$ in $\Phi_{\ASPEC}(s_1,s_1) = {\pi}_0 \wedge {\pi}_1 \wedge \pi_2 \wedge \pi_3$.
On the graph $G_{\Ac}$, this path corresponds to $\widetilde{\Lambda}(e_1e_3) = \{\tupleof{\{\pi_0, \pi_2\},(s_1,s_1),0}\}$.
Similarly, in order to enable the path $(q_0,s_1),(q_1,s_1),(q_2,s_1)$ on the product automaton, we need to replace $\pi_0$, $\pi_2$  and $\pi_3$ with $\top$ in $\Phi_{\ASPEC}(s_1,s_1)$.
On the graph $G_{\Ac}$, this path corresponds to $\widetilde{\Lambda}(e_1e_2) = \{ \tupleof{\{\pi_0, \pi_2\, \pi_3\},(s_1,s_1),0}\}$.
} 

\yhl{
Therefore, the path defined by edges $e_1$ and $e_3$ is preferable over the path defined by edges $e_1$ and $e_2$.
In the first case, we have cost $C(e_1e_3) = 2$ which corresponds to relaxing 2 requirements, i.e., $\pi_0$ and $\pi_2$, while in the latter case, we have cost $C(e_1e_2) = 3$ which corresponds to relaxing 3 requirements, i.e., $\pi_0$, $\pi_2$  and $\pi_3$.
}\exmend
\label{exmp:simple:nota}
\end{exmp}

\ifthenelse {\boolean{BGRAPHPDF}}
{
\begin{figure}[t]
\centering

\includegraphics[width=8cm]{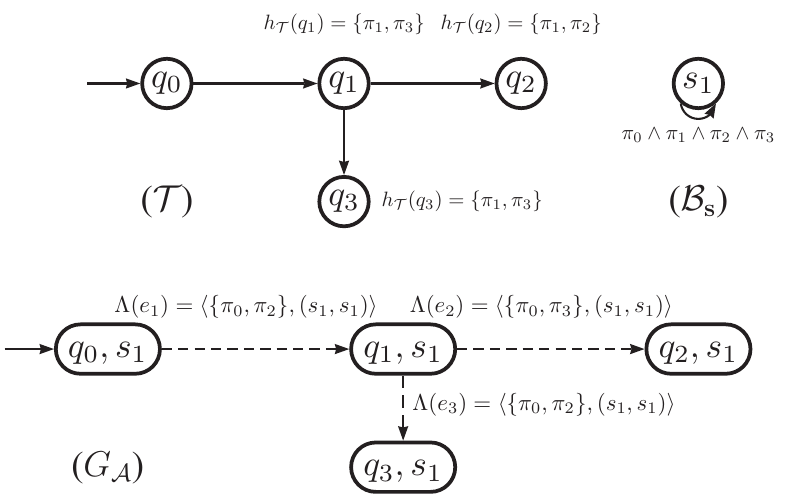}
\caption{Example \ref{exmp:simple:nota}.
$\FTS$: part of the system;
$\ASPEC$: part of the specification automaton;
$G_\Ac$: part of the graph that corresponds to the product automaton.
}

\label{fig:product}
\end{figure}
}
{
\begin{figure}
\begin{center}
\VCDraw{%
\begin{VCPicture}{(-1,-3)(7,2.6)}
\rput(-1,-1){\LARGE($\FTS$)}
\rput(8,-1){\LARGE($\ASPEC$)}
\State[q_0]{(-1,1)}{A} 
\State[q_1]{(2,1)}{B} 
\rput(2,2){$h_\FTS(q_1)=\{\pi_1,\pi_3\}$}
\State[q_2]{(5,1)}{C} 
\rput(5,2){$h_\FTS(q_2)=\{\pi_1,\pi_2\}$}
\State[q_3]{(2,-1)}{D} 
\rput(4,-1){$h_\FTS(q_3)=\{\pi_1,\pi_3\}$}
\State[s_1]{(8,1)}{E} 
\Initial{A}  
\ncline{->}{A}{B} 
\ncline{->}{B}{C} 
\ncline{->}{B}{D} 
\nccircle{->}{E}{-0.3cm}\naput{${\pi}_0 \wedge {\pi}_1 \wedge \pi_2 \wedge \pi_3$}
\end{VCPicture}}\\
\VCDraw{%
\begin{VCPicture}{(-1,-2)(9,1.5)}
\rput(-1,-1){\LARGE($G_{\Ac}$)}
\StateVar[{q_0, s_1}]{(-1,1)}{A} \StateVar[q_1, s_1]{(4,1)}{B} \StateVar[q_2, s_1]{(9,1)}{C} \StateVar[q_3, s_1]{(4,-1)}{D}
\Initial{A}
\ChgEdgeLineStyle{dashed} 
\ChgEdgeLabelScale{0.65}
\EdgeL{A}{B} {\StackTwoLabels{\Lambda(e_1) = \tupleof{\{\pi_0, \pi_2\},(s_1,s_1)}}{\;}}
\EdgeL{B}{C} {\StackTwoLabels{\Lambda(e_2) = \tupleof{\{\pi_0, \pi_3\},(s_1,s_1)}}{\;}}
\EdgeL{B}{D} {\Lambda(e_3) = \tupleof{\{\pi_0, \pi_2\},(s_1,s_1)}}  
\end{VCPicture}}
\caption{Example \ref{exmp:simple:nota}.
$\FTS$: part of the system;
$\ASPEC$: part of the specification automaton;
$G_\Ac$: part of the graph that corresponds to the product automaton.
}
\end{center}
\label{fig:product}
\end{figure}

}

\yhl{
A valid relaxation $\BUCHI$ should produce a reachable $v_f \in V_f$ with prefix and lasso path such that $\Lc(\FTS \times \BUCHI) \neq \emptyset$.
The next section provides an algorithmic solution to this problem.
}

%

%% file: heuristic.tex
\section{A Heuristic Algorithm for MRP}
\label{sec:heuristic}

In this section, we present an approximation algorithm (AAMRP) for the Minimal Revision Problem (MRP). 
It is based on Dijkstra's shortest path algorithm (\citealt{CormenLRS01}). 
The main difference from Dijkstra's algorithm is that instead of finding the minimum weight path to reach each node, AAMRP tracks the number of atomic propositions that must be removed from each edge on the paths of the graph $G_\Ac$. 

The pseudocode for the AAMRP is presented in Algorithms \ref{alg:aamrp:main} and \ref{alg:aamrp:sub}.
The main algorithm (Alg. \ref{alg:aamrp:main}) divides the problem into two tasks.
First, in Line \ref{alg:main:prefix}, it finds an approximation to the minimum number of atomic propositions from $\APRem$ that must be removed to have a prefix path to each reachable sink (see Section \ref{sec:ltlplanning}).
Then, in Line \ref{alg:main:lasso}, it repeats the process from each reachable final state to find an approximation to the minimum number of atomic propositions that must be removed so that a lasso path is enabled.
The combination of prefix/lasso that removes the minimal number of atomic propositions is returned to the user. \yhl{We remark that from line \ref{alg:main:store}, a set of atomic propositions found from prefix part is used when it starts searching for lasso path of every reachable $v_f \in \Vc \cap V_f$.}

\ifthenelse {\boolean{BGRAPHPDF}}
{
\begin{algorithm}[tb]
{\bf Inputs}: a graph $G_\Ac = (V,E,v_s,V_f,\APRem,\Lambda)$. \\
{\bf Outputs}: the list $L$ of atomic propositions form $\APRem$ that must be removed $\ASPEC$. 
\caption{AAMRP}
\label{alg:aamrp:main}
\begin{algorithmic}[1]
\Procedure{AAMRP}{$G_\Ac$}
\State $L \gets \APRem$
\State $\Mc[:,:] \gets (\APRem,\infty)$
\State $\Mc[v_s,:] \gets (\emptyset,0)$ \Comment{Initialize the source node}
\State $\tupleof{\Mc, \Pb, \Vc} \gets  \text{\sc FindMinPath}(G_\Ac,\Mc,0)$ \label{alg:main:prefix}
\State \yhl{$\text{\sc Acceptable} \gets False$}
\For {$v_f \in \Vc \cap V_f$} \label{alg:main:loop}
\State $L_p \gets \text{\sc GetAPFromPath}(v_s,v_f,\Mc,\Pb)$  \label{alg:main:lasso}
\State $\Mc'[:,:] \gets (\APRem,\infty)$
\State \yhl{$\Mc'[v_f,:] \gets (L_p,|L_p|)$} \label{alg:main:store} \Comment{Store APs from prefix path $v_s \leadsto v_f$ to $\Mc'[v_f,:]$}
\State $G_\Ac' \gets (V,E,v_f,\{v_f\},\APRem,\Lambda)$
\State $\tupleof{\Mc', \Pb', \Vc'} \gets  \text{\sc FindMinPath}(G'_\Ac,\Mc',1)$
\If {$v_f \in \Vc'$} \label{alg:main:cycled}
\State \yhl{$L' \gets \text{\sc GetAPFromPath}(v_f,v_f,\Mc',\Pb')$} \Comment{Get APs of prefix $v_s \leadsto v_f$ and lasso $v_f \leadsto v_f$ from $\Mc'[v_f,:]$}
\If {\yhl{$|L'|$} $\leq |L|$} \label{alg:main:less}
\State $L \gets L'$
\EndIf
\State \yhl{$\text{\sc Acceptable} \gets True$}
\EndIf
\EndFor
\If {\yhl{$\neg \text{\sc Acceptable}$}} \label{alg:main:unreachable}
\State $L \gets \emptyset$
\EndIf
\State \Return $L$
\EndProcedure
\end{algorithmic}
The function \text{\sc GetAPFromPath}($(v_s,v_f,\Mc,\Pb)$) returns the atomic propositions that must be removed from $\ASPEC$ in order to enable a path on $\Ac$ from a starting state $v_s$ to a final state $v_f$ given the tables $\Mc$ and $\Pb$.
\end{algorithm}
}
{
}

Algorithm \ref{alg:aamrp:sub} follows closely Dijkstra's shortest path algorithm (\citealt{CormenLRS01}). 
It maintains a list of visited nodes $\Vc$ and a table $\Mc$ indexed by the graph vertices which stores the set of atomic propositions that must be removed in order to reach a particular node on the graph.
Given a node $v$, the size of the set $|\Mc[v,1]|$ is an upper bound on the minimum number of atomic propositions that must be removed.
That is, if we remove all $\overline{\pi} \in \Mc[v,1]$ from $\ASPEC$, then we enable a simple path (i.e., with no cycles) from a starting state to the state $v$. 
The size of $|\Mc[v,1]|$ is stored in $\Mc[v,2]$ which also indicates that the node $v$ is reachable when $\Mc[v,2] < \infty$.

The algorithm works by maintaining a queue with the unvisited nodes on the graph.
Each node $v$ in the queue has as key the number of atomic propositions that must be removed so that $v$ becomes reachable on $\Ac$.
The algorithm proceeds by choosing the node with the minimum number of atomic propositions discovered so far (line \ref{alg:aamrp:sub:min}).
Then, this node is used in order to updated the estimates for the minimum number of atomic propositions needed in order to reach its neighbors (line \ref{alg:aamrp:sub:relax}).
A notable difference of Alg. \ref{alg:aamrp:sub} from Dijkstra's shortest path algorithm is the check for lasso paths in lines \ref{alg:aamrp:sub:1}-\ref{alg:aamrp:sub:end}.
After the source node is used for updating the estimates of its neighbors, its own estimate for the minimum number of atomic propositions is updated either to the value indicated by the self loop or the maximum possible number of atomic propositions.
This is required in order to compare the different paths that reach a node from itself.

\input{algorithms}

\ifthenelse {\boolean{BGRAPHPDF}}
{
\begin{figure}[t]
\centering
\includegraphics[width=8cm]{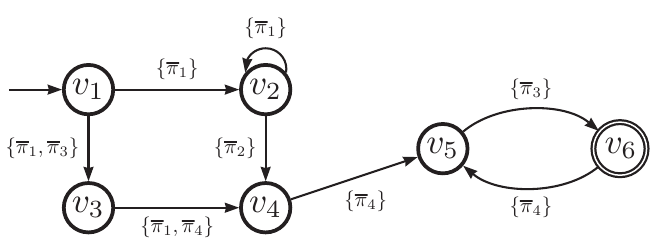}
\caption{The graph of Example \ref{exmp:heur:01}. The source $v_s = v_1$ is denoted by an arrow and the sink $v_6$ by double circle ($V_f = \{v_6\}$).}
\label{fig:exmp:heu:graph}
\end{figure}
}
{
\begin{figure}
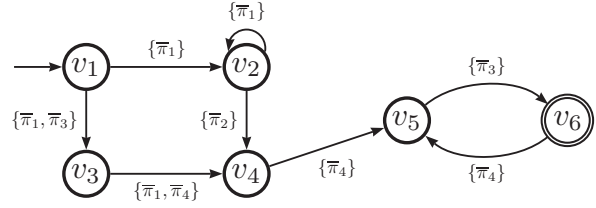

\begin{center}
\VCDraw{%
\begin{VCPicture}{(0,-1.7)(9,2)}
\State[v_1]{(0,1)}{A} \State[v_2]{(3,1)}{B} \State[v_3]{(0,-1)}{C}
\State[v_4]{(3,-1)}{D} \State[v_5]{(6,0)}{E} \FinalState[v_6]{(9,0)}{F}
\Initial{A}  
\ncline{->}{A}{B} \naput{$\{\overline{\pi}_1\}$}
\ncline{->}{A}{C} \nbput{$\{\overline{\pi}_1,\overline{\pi}_3\}$}
\ncline{->}{C}{D} \nbput{$\{\overline{\pi}_1,\overline{\pi}_4\}$}
\ncline{->}{B}{D} \nbput{$\{\overline{\pi}_2\}$}
\ncline{->}{D}{E} \nbput{$\{\overline{\pi}_4\}$}
\ncarc[arcangle=40]{->}{E}{F} \naput{$\{\overline{\pi}_3\}$}
\ncarc[arcangle=40]{->}{F}{E} \naput{$\{\overline{\pi}_4\}$}
\nccircle{->}{B}{0.35cm}\nbput{$\{\overline{\pi}_1\}$}
\end{VCPicture}}
\end{center}
\caption{The graph of Example \ref{exmp:heur:01}. The source $v_s = v_1$ is denoted by an arrow and the sink $v_6$ by double circle ($V_f = \{v_6\}$).}
\label{fig:exmp:heu:graph}
\end{figure}
}

The following example demonstrates how the algorithm works and indicates the structural conditions on the graph that make the algorithm non-optimal.

\begin{exmp}
\label{exmp:heur:01}
Let us consider the graph in Fig. \ref{fig:exmp:heu:graph}. 
The source node of this graph is $v_s = v_1$ and the set of sink nodes is $V_f = \{v_6\}$. 
The $\APRem$ set of this graph is $\{\aprem_1,\ldots,\aprem_4\}$.
Consider the first call of {\sc FindMinPath} (line \ref{alg:main:prefix} of Alg. \ref{alg:aamrp:main}).

\begin{itemize}
\item Before the first execution of the while loop (line \ref{alg:aamrp:sub:while}): 
The queue contains $\Qc = \{v_2, \ldots, v_6 \}$.
The table $\Mc$ has the following entries: $\Mc{[v_1,:]} = \tupleof{\emptyset,0}$, $\Mc{[v_2,:]} = \tupleof{\{\aprem_1\},1}$, $\Mc{[v_3,:]} = \tupleof{\{\aprem_1, \aprem_3\},2}$, $\Mc{[v_4,:]} = \ldots = \Mc{[v_6,:]} = \tupleof{\APRem,\infty}$. 

\item Before the second execution of the while loop (line \ref{alg:aamrp:sub:while}): 
The node $v_2$ was popped from the queue since it had $\Mc{[v_2,2]} = 1$.
The queue now contains $\Qc = \{v_3, \ldots, v_6 \}$.
The table $\Mc$ has the following rows: $\Mc{[v_1,:]} = \tupleof{\emptyset,1}$, $\Mc{[v_2,:]} = \tupleof{\{\aprem_1\},1}$, $\Mc{[v_3,:]} = \tupleof{\{\aprem_1, \aprem_3\},2}$, $\Mc{[v_4]} = \tupleof{\{\aprem_1, \aprem_2\},2}$, $\Mc{[v_5,:]} = \Mc{[v_6,:]} = \tupleof{\APRem,\infty}$. 

\item At the end of {\sc FindMinPath} (line \ref{alg:aamrp:sub:endproc}): 
The queue now is empty.
The table $\Mc$ has the following rows: $\Mc{[v_1,:]} = \tupleof{\emptyset,0}$, $\Mc{[v_2,:]} = \tupleof{\{\aprem_1\},1}$, $\Mc{[v_3,:]} = \tupleof{\{\aprem_1, \aprem_3\},2}$, $\Mc{[v_4,:]} = \tupleof{\{\aprem_1, \aprem_2\},2}$, $\Mc{[v_5,:]} = \tupleof{\{\aprem_1, \aprem_2, \aprem_4\},3}$, $\Mc{[v_6,:]} = \tupleof{\APRem,4}$, which corresponds to the path $v_1$, $v_2$, $v_4$, $v_5$, $v_6$.

\end{itemize}

Note that algorithm returns a set of atomic propositions $\yhl{L'} = \APRem$ which is not optimal \yhl{$|L'|$} $= 4$.
 The path $v_1$, $v_3$, $v_4$, $v_5$, $v_6$ would return \yhl{$L'$} $= \{\aprem_1, \aprem_3, \aprem_4\}$ with \yhl{$|L'|$} $= 3$.
\exmend
\end{exmp}

{\bf Correctness:} 
The correctness of the algorithm AAMRP is based upon the fact that a node $v \in V$ is reachable on $G_\Ac$ if and only if $\Mc[v,2] < \infty$.
The argument for this claim is similar to the proof of correctness of Dijkstra's shortest path algorithm in \citealt{CormenLRS01}.
If this algorithm returns a set of atomic propositions $L$ which removed from $\ASPEC$, then the language $\Lc(\Ac)$ is non-empty. 
This is immediate by the construction of the graph $G_{\Ac}$ (Def. \ref{def:graph}).


\yhl{We remark that AAMRP does not solve Problem \ref{prb:formal} exactly since MRP is NP-Complete. 
However, AAMRP guarantees that it returns a valid relaxation $\BUCHI$ where $\ASPEC \preceq \BUCHI$.}

\yhl{
\begin{thm}
If a valid relaxation exists, then AAMRP always returns a valid relaxation $\BUCHI$ of some initial $\ASPEC$ such that $\Lc(\FTS \times \BUCHI) \neq \emptyset$.
\label{thm:valid_relaxation}
\end{thm}
}

\begin{proof}
\yhl{First, we will show that if AAMRP returns $\emptyset$, then there is no valid relaxation of $\ASPEC$.
AAMRP returns $\emptyset$ when there is no reachable $v_f \in V_f$ with prefix and lasso path or {\sc GetAPFromPath} returns $\emptyset$. 
If there is no reachable $v_f$, then either the accepting state is not reachable on $\ASPEC$ or on $\FTS$.
Recall that the Def. \ref{def:graph} constructs a graph where all the transitions of $\FTS$ and $\BUCHI$ are possible.
If it returns $\emptyset$ as a valid solution, then there is a path on the graph that does not utilize any labeled edge by $\Lambda$.
Thus, $\Lc(\FTS \times \ASPEC) \neq \emptyset$.
Since we assume that $\ASPEC$ is unsatisfiable on $\FTS$, this is contradiction.
}

\yhl{Second, without loss of generality, suppose that AAMRP returns $\widetilde{\Lambda}(\mu)$.
Using this $\widetilde{\Lambda}(\mu)$, we can build a relax specification automaton $\BUCHI$.
Using each $\tupleof{l,(s,s'),k} \in \widetilde{\Lambda}(\mu)$ and for each $\pi \in l$, we add the indices of the literal $\phi_{ij}$ in $\Phi_{\ASPEC(s,s')}$ that corresponds to $\pi$ to the sets $\hat D_{ss'}$ and $\hat C^i_{ss'}$.
The resulting substitution $\theta$ produces a relaxation.
Moreover, it is a valid relaxation, because by removing the atomic propositions in $\theta$ from $\ASPEC$, we get a path that satisfies the prefix and lasso components on the product automaton.
}
\end{proof}

{\bf Running time:} 
The running time analysis of the AAMRP is similar to that of Dijkstra's shortest path algorithm.
In the following, we will abuse notation when we use the $O$ notation and treat each set symbol $S$ as its cardinality $|S|$.

First, we will consider {\sc FindMinPath}.
The fundamental difference of AAMRP over Dijkstra's algorithm is that we have set theoretic operations.
We will assume that we are using a data structure for sets that supports $O(1)$ set cardinality quarries, $O(\log n)$ membership quarries and element insertions (\citealt{CormenLRS01}) and $O(n)$ set up time.
Under the assumption that $\Qc$ is implemented in such a data structure, each {\sc ExtractMIN} takes $O(\log V)$ time.
Furthermore, we have $O(V)$ such operations (actually $|V|-1$) for a total of $O(V \log V)$.

Setting up the data structure for $\Qc$ will take $O(V)$ time.
Furthermore, in the worst case, we have a set $\Lambda(e)$ for each edge $e \in E$ with set-up time $O(E \APRem)$.
Note that the initialization of $\Mc[v,:]$ to $\tupleof{\APRem,\infty}$ does not have to be implemented since we can have indicator variables indicating when a set is supposed to contain all the (known in advance) elements.

Assuming that $E$ is stored in an adjacency list, the total number of calls to {\sc Relax} at lines \ref{alg:aamrp:sub:for1} and \ref{alg:aamrp:sub:for2} of Alg. \ref{alg:aamrp:sub} will be $O(E)$ times.
Each call to {\sc Relax} will have to perform a union of two sets ($\Mc[u,1]$ and $\Lambda(u,v)$).
Assuming that both sets have in the worst case $|\APRem|$ elements, each union will take $O(\APRem \log \APRem)$ time.
Finally, each set size quarry takes $O(1)$ time and updating the keys in $\Qc$ takes $O(\log V)$ time.
Therefore, the running time of  {\sc FindMinPath} is $O(V + E \APRem + V \log V + E (\APRem \log \APRem + \log V))$.

Note that even if under Assumption \ref{ass:reach} all nodes of $\FTS$ are reachable $(|V| < |E|)$, the same property does not hold for the product automaton. (e.g, think of an environment $\FTS$ and a specification automaton whose graphs are Directed Acyclic Graphs (DAG). However, even in this case, we have $(|V| < |E|)$.
The running time of {\sc FindMinPath} is  $O(E (\APRem \log \APRem + \log V))$.
Therefore, we observe that the running time also depends on the size of the set $\APRem$.
However, such a bound is very pessimistic since not all the edges will be disabled on $\Ac$ and, moreover, most edges will not have the whole set $\APRem$ as candidates for removal.

Finally, we consider {\sc AAMRP}.
The loop at line \ref{alg:main:loop} is going to be called $O(V_f)$ times.
At each iteration, {\sc FindMinPath} is called.
Furthermore, each call to {\sc GetAPFromPath} is going to take $O(V \APRem \log \APRem)$ time (in the worst case we are going to have $|V|$ unions of sets of atomic propositions).
Therefore, the running time of {\sc AAMRP} is $O(V_f ( V \APRem \log \APRem + E (\APRem \log \APRem + \log V))) = O(V_f E (\APRem \log \APRem + \log V))$ which is polynomial in the size of the input graph.

{\bf Approximation bound:}
AAMRP does not have a constant approximation ratio on arbitrary graphs.

\ifthenelse {\boolean{BGRAPHPDF}}
{
\begin{figure}[t]
\centering
\includegraphics[width=8cm]{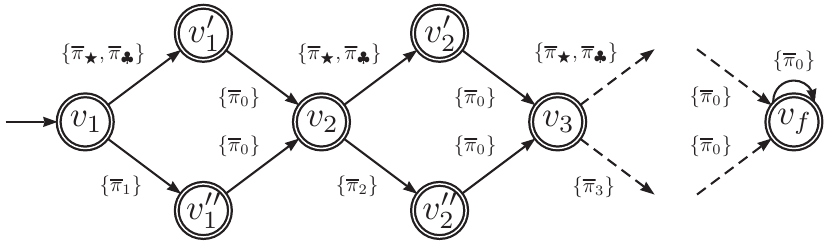}
\caption{The graph of Example \ref{exmp:heur:03}. The source $v_s = v_1$ is denoted by an arrow and the sink $v_f$ by double circle ($V_f = \{v_f\}$).}
\label{fig:exmp:heu:graph3}
\end{figure}
}
{
\begin{figure}
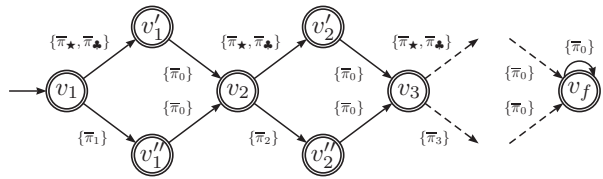

\begin{center}
\VCDraw{%
\begin{VCPicture}{(-1,-1.5)(12,1.5)}
\FinalState[v_1]{(0,0)}{A1}
\FinalState[v_1']{(2,1.5)}{B1}
\FinalState[v_1'']{(2,-1.5)}{C1}
\FinalState[v_2]{(4,0)}{A2}
\FinalState[v_2']{(6,1.5)}{B2}
\FinalState[v_2'']{(6,-1.5)}{C2}
\FinalState[v_3]{(8,0)}{A3}
\HideState
\State[v_m']{(10,1.5)}{Bm}
\State[v_m'']{(10,-1.5)}{Cm}
\ShowState
\FinalState[v_f]{(12,0)}{Af}
\Initial{A1}
\ncline{->}{A1}{B1} \naput{$\{\overline{\pi}_\bigstar,\overline{\pi}_\clubsuit\}$}
\ncline{->}{A1}{C1} \nbput{$\{\overline{\pi}_1\}$}
\ncline{->}{B1}{A2} \nbput{$\{\overline{\pi}_0\}$}
\ncline{->}{C1}{A2} \naput{$\{\overline{\pi}_0\}$}
\ncline{->}{A2}{B2} \naput{$\{\overline{\pi}_\bigstar,\overline{\pi}_\clubsuit\}$}
\ncline{->}{A2}{C2} \nbput{$\{\overline{\pi}_2\}$}
\ncline{->}{B2}{A3} \nbput{$\{\overline{\pi}_0\}$}
\ncline{->}{C2}{A3} \naput{$\{\overline{\pi}_0\}$}
\nccircle{<-}{Af}{0.35cm}\nbput{$\{\overline{\pi}_0\}$}
\ChgEdgeLineStyle{dashed}
\ncline{->}{A3}{Bm} \naput{$\{\overline{\pi}_\bigstar,\overline{\pi}_\clubsuit\}$}
\ncline{->}{A3}{Cm} \nbput{$\{\overline{\pi}_3\}$}
\ncline{->}{Bm}{Af} \nbput{$\{\overline{\pi}_0\}$}
\ncline{->}{Cm}{Af} \naput{$\{\overline{\pi}_0\}$}
\end{VCPicture}}
\end{center}
\caption{The graph of Example \ref{exmp:heur:03}. The source $v_s = v_1$ is denoted by an arrow and the sink $v_f$ by double circle ($V_f = \{v_f\}$).}
\label{fig:exmp:heu:graph3}
\end{figure}
}

\begin{exmp}[Unbounded Approximation]
\label{exmp:heur:03}
The graph in Fig. \ref{fig:exmp:heu:graph3} is the product of a specification automaton with a single state and a self transition with label $\{\overline{\pi}_0,\overline{\pi}_1,\ldots,\overline{\pi}_m,\overline{\pi}_\bigstar,\overline{\pi}_\clubsuit\}$ and an environment automaton with the same structure as the graph in Fig. \ref{fig:exmp:heu:graph3} but with appropriately defined state labels.
In this graph, AAMRP will choose the path $v_1$,$v_1''$, $v_2$, $v_2''$, $v_3$, $\ldots$, $v_f$. 
The corresponding revision will be the set of atomic propositions $L_p = \{\overline{\pi}_0, \overline{\pi}_1, \overline{\pi}_2, \ldots, \overline{\pi}_m\}$ with $|L_p| = m+1$.
This is because in $v_2$, AAMRP will choose the path through $v_1''$ rather than $v_1'$ since the latter will produce a revision set of size $|\{\overline{\pi}_0, \overline{\pi}_\bigstar,\overline{\pi}_\clubsuit\}| =3$ while the former a revision set of size $|\{\overline{\pi}_0, \overline{\pi}_1\}| =2$.
Similarly at the next junction node $v_3$, the two candidate revision sets $\{\overline{\pi}_0, \overline{\pi}_1, \overline{\pi}_\bigstar,\overline{\pi}_\clubsuit\}$ and $\{\overline{\pi}_0, \overline{\pi}_1, \overline{\pi}_2\}$ have sizes 4 and 3, respectively.
Therefore, the algorithm will always choose the path through the nodes $v_i''$ rather than $v_i'$ producing, thus, a solution of size $m+1$. 
However, in this graph, the optimal revision would have been $L_p = \{\overline{\pi}_0, \overline{\pi}_\bigstar,\overline{\pi}_\clubsuit\}$ with $|L_p| = 3$.
Hence, we can see that in this example for $m>2$ AAMRP returns a solution which is $m-2$ times bigger than the optimal solution.
\exmend
\end{exmp}

There is also a special case where AAMRP returns a solution whose size is at most twice the size of the optimal solution.

\begin{thm}
AAMRP on planar Directed Acyclic Graphs (DAG) where all the paths merge on the same node is a 2-approximation algorithm.
\end{thm}

The proof is provided in the Appendix \ref{app:bound}.

%% file: algorithms.tex
\ifthenelse {\boolean{BGRAPHPDF}}
{
\begin{algorithm}[tb]
\caption{{\sc FindMinPath}}
{\bf Inputs}: a graph $G_\Ac = (V,E,v_s,V_f,\APRem,\Lambda)$, a table
$\Mc$ and a flag $lasso$ on whether this is a lasso path search. \\
{\bf Variables}: a queue $\Qc$, a set $\Vc$ of visited nodes and a
table $\mathbf{P}$ indicating the parent of each node on a path. \\
{\bf Output}: the tables $\Mc$ and $\mathbf{P}$ and the visited nodes $\Vc$
\label{alg:aamrp:sub}
\begin{algorithmic}[1]
\Procedure{{\sc FindMinPath}}{$G_\Ac$,$\Mc$,$lasso$}
\State $\Vc \gets \{v_s\}$
\State $\mathbf{P}[:]  \gets \emptyset$ \Comment{Each entry of
$\mathbf{P}$ is set to $\emptyset$}
\State $\Qc \gets V - \{v_s\}$
\For {$v \in V$ such that $(v_s,v) \in E$ and $v \neq v_s$}
\label{alg:aamrp:sub:for1}
\State $\tupleof{\Mc,\Pb} \gets \text{\sc Relax}((v_s,v),\Mc,\Pb,\Lambda)$
\EndFor
\If {$lasso=1$} \label{alg:aamrp:sub:1}
\If {$(v_s,v_s) \in E$}
\State $\Mc[v_s,1] \gets \Mc[v_s,1] \cup \Lambda(v_s,v_s)$
\State $\Mc[v_s,2] \gets |\Mc[v_s,1] \cup \Lambda(v_s,v_s)|$
\State $\Pb[v_s] = v_s$
\Else
\State $\Mc[v_s,:] \gets (\APRem,\infty)$
\EndIf
\EndIf \label{alg:aamrp:sub:end}
\While {$\Qc \neq \emptyset$} \label{alg:aamrp:sub:while}
\State $u \gets$ {\sc ExtractMIN}($\Qc$) \label{alg:aamrp:sub:min} \Comment{Get node $u$ with minimum $\Mc[u,2]$}
\If {$\Mc[u,2] < \infty$}
\State $\Vc \gets \Vc \cup \{u\}$
\For {$v \in V$ such that $(u,v) \in E$} \label{alg:aamrp:sub:for2}
\State $\tupleof{\Mc,\Pb} \gets \text{\sc
Relax}((u,v),\Mc,\Pb,\Lambda)$ \label{alg:aamrp:sub:relax}
\EndFor
\EndIf
\EndWhile
\State \Return $\Mc$, $\mathbf{P}$, $\Vc$
\EndProcedure \label{alg:aamrp:sub:endproc}
\end{algorithmic}
\end{algorithm}
}
{
}

\ifthenelse {\boolean{BGRAPHPDF}}
{
\begin{algorithm}[tb]
\caption{{\sc Relax}}
{\bf Inputs}: an edge $(u,v)$, the tables $\Mc$ and $\Pb$ and the edge labeling function $\Lambda$ \\
{\bf Output}: the tables $\Mc$ and $\mathbf{P}$
\label{alg:aamrp:relax}
\begin{algorithmic}[1]
\Procedure{{\sc Relax}}{$(u,v)$,$\Mc$,$\Pb$,$\Lambda$}
\If {$|\Mc[u,1] \cup \Lambda(u,v)| < \Mc[v,2]$}
\State $\Mc[v,1] \gets \Mc[u,1] \cup \Lambda(u,v)$
\State $\Mc[v,2] \gets |\Mc[u,1] \cup \Lambda(u,v)|$
\State $\mathbf{P}[v] \gets u$
\EndIf
\State \Return $\Mc$, $\mathbf{P}$
\EndProcedure
\end{algorithmic}
\end{algorithm}
}
{
}

%% file: experiments.tex
\section{Examples and Numerical Experiments}
\label{sec:experiments}

In this section, we present experimental results using our prototype implementation of AAMRP. 
The prototype implementation is in Python (see \citealt{SRPT}).
Therefore, we expect the running times to substantially improve with a C implementation using state-of-the-art data structure implementations.

We first present some examples and expand few more example scenarios.

\begin{figure}[t]
\centering
\includegraphics[width=7cm]{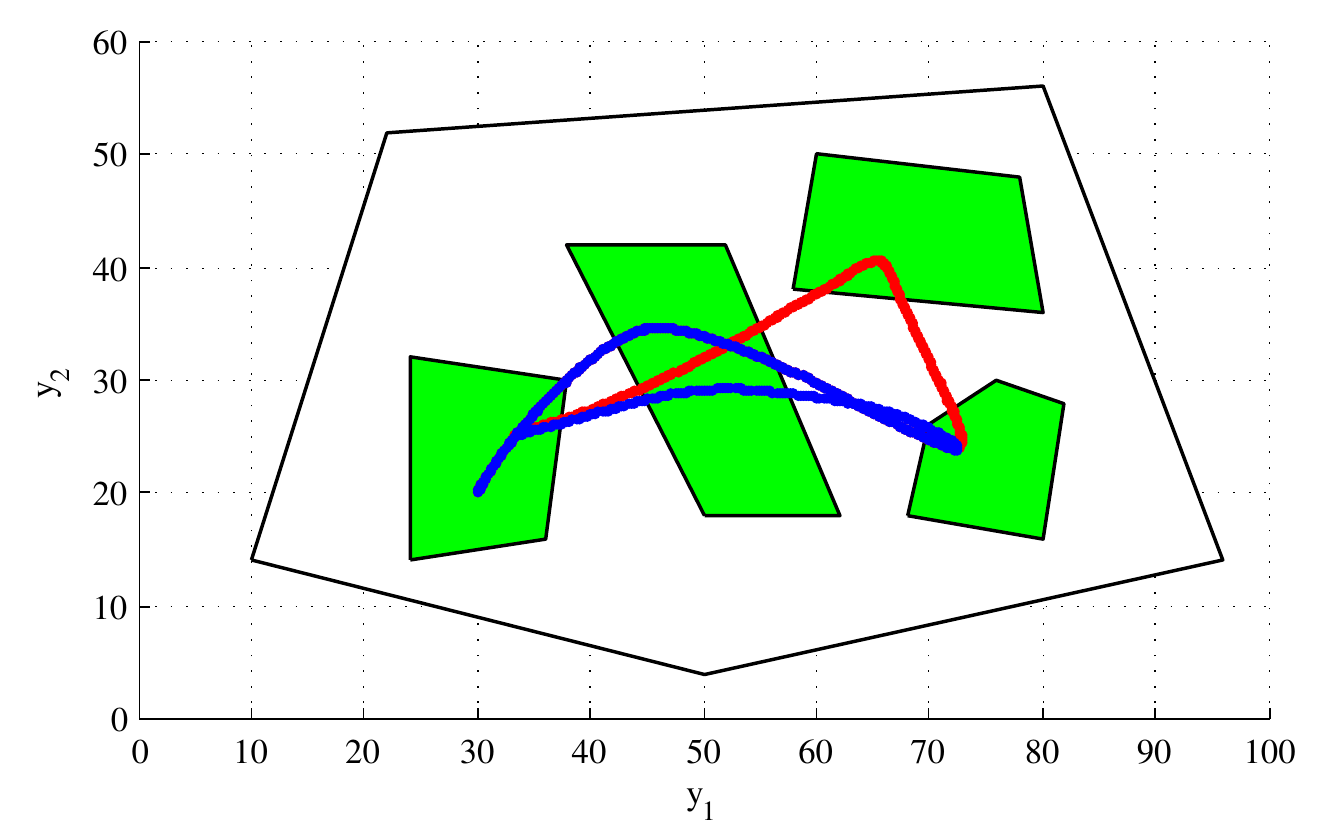}
\caption{The simple environment of Example \ref{exm:robot_motion_planning} along with a low speed mobile robot trajectory that satisfies the specification.}
\label{fig:rev:traj}
\end{figure}

\begin{exmp}
\label{exm:robot_motion_planning}
We revisit Example \ref{exm:areas4}. The product automaton of this example has 85 states, 910 transitions and 17 reachable final states.
It takes 0.095 sec by AAMRP.
AAMRP returns the set of atomic propositions $\{ \tupleof{\pi_{4}, (s_{2},s_{4})}\}$ as a minimal revision to the problem, which is revision (3) among the three minimal revisions of the example: one of the blue trajectories in Fig. \ref{fig:rev:traj}.
Thus, it is an optimal solution.
\exmend
\end{exmp}

\begin{exmp}
We revisit Example \ref{exm:twoagents}.
The graph of this example has 36 states, 240 transitions and 9 reachable sinks.
AAMRP returns the set of atomic propositions $\{ \tupleof{\pi_{13}, (s_2,s_3) }\}$ as minimal revision to the problem.
It takes 0.038 sec by AAMRP.
Intuitively, AAMRP recommends dropping the requirement that $\pi_{13}$ should be reached from the specification.
Therefore, Object 1 will remain where it is, while Object 2 will follow the path $q_1$, $q_2$, $q_3$, $q_2$, $q_3$, \ldots.
\exmend
\end{exmp}

With our prototype implementation, we could expand our experiment to few more example scenarios introduced in \citealt{UlusoyEtAl2011iros, UluosoySDB2012}.

\begin{figure}[t]
\centering
\includegraphics[width=7cm]{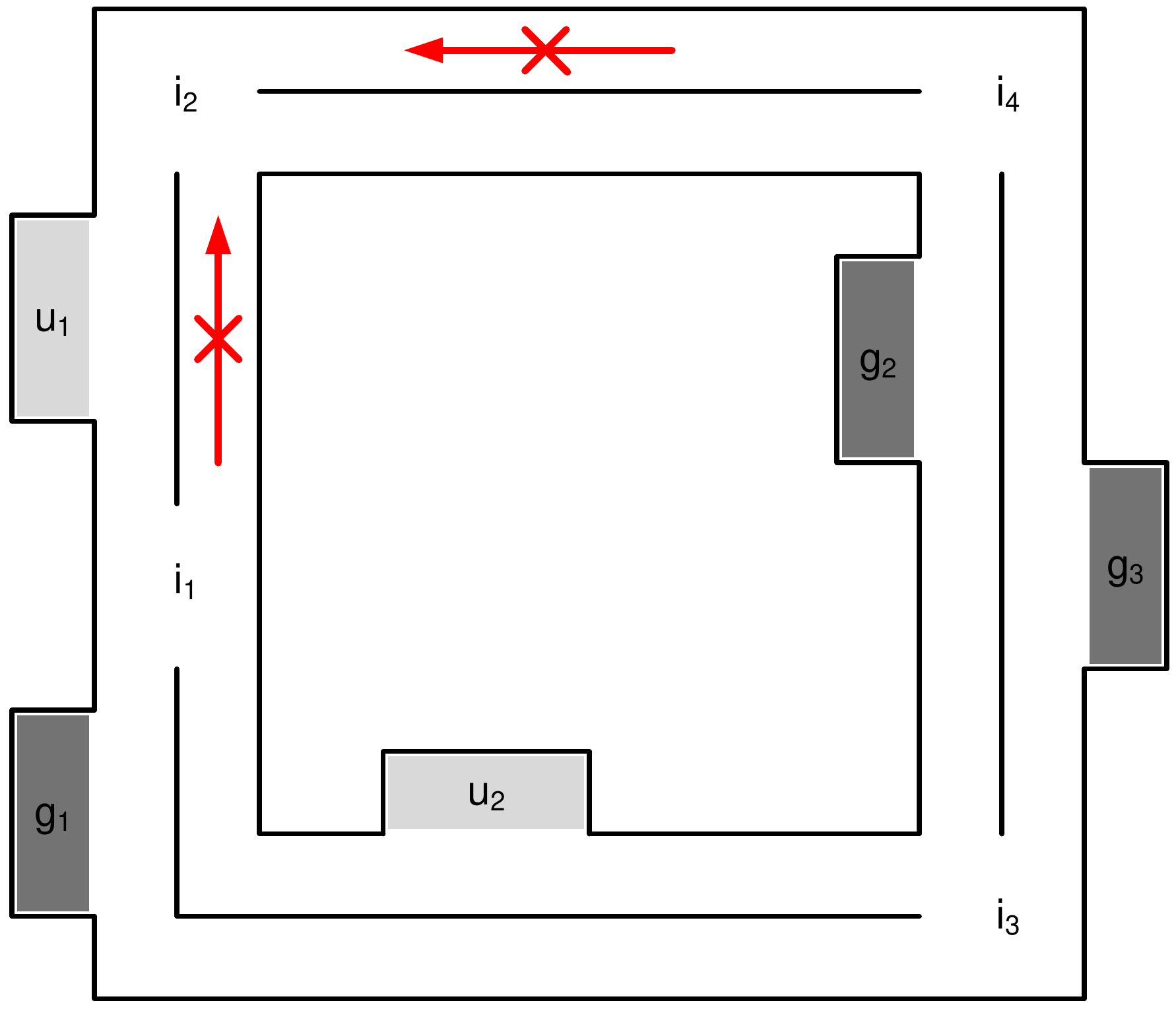}
\caption{Schematic illustration of the simple road network environment of Example \ref{exm:single_robot_gathering}. The robot is required to drive right-side of the road.}
\label{pic:single_robot_road}
\end{figure}

\begin{exmp}[Single Robot Data Gathering Task]
\label{exm:single_robot_gathering}

In this example, we use a simplified road network having three gathering locations and two upload locations with four intersections of the road. In Fig. \ref{pic:single_robot_road}, the data gather locations, which are labeled $g_1$, $g_2$, and $g_3$, are dark gray, the data upload locations, which are labeled $u_1$ and $u_2$, are light gray, and the intersections are labeled $i_1$ through $i_4$.
In order to gather data and upload the gather-data persistently, the following LTL formula may be considered:
$\phi_A$ $:=$ GF($\varphi_g$) $\wedge$ GF($\pi$), where $\varphi_g := g_1 \vee g_2 \vee g_3$ and $\pi := u_1 \vee u_2$.
The following formula can make the robot move from gather locations to upload locations after gathering data:
$\phi_G$ $:=$ G($\varphi_g \rightarrow$ X($\neg\varphi_g \Un \pi$).
In order for the robot to move to gather location after uploading, the following formula is needed:
$\phi_U$ $:=$ G($\pi \rightarrow$ X($\neg\pi \Un \varphi_g$).

Let us consider that some parts of road are not recommended to drive from gather locations, such as from $i_4$ to $i_2$ and from $i_1$ to $i_2$.
We can describe those constraints as following:
$\psi_1$ $:=$ G($g_1$ $\rightarrow$ $\neg$($i_4$ $\wedge$ X$i_2$)$\Un$$u_1$) and
$\psi_2$ $:=$ G($g_2$ $\rightarrow$ $\neg$($i_1$ $\wedge$ X$i_2$)$\Un$$u_2$).
If the gathering task should have an order such as $g_3$, $g_1$, $g_2$, $g_3$, $g_1$, $g_2$, $\ldots$, then the following formula could be considered:
$\phi_O$ := 
(($\neg$$g_1$ $\wedge$ $\neg$$g_2$)$\Un$$g_3$) $\wedge$
G($g_3$ $\rightarrow$ X(($\neg$$g_2$ $\wedge$ $\neg$$g_3$)$\Un$$g_1$)) $\wedge$
G($g_1$ $\rightarrow$ X(($\neg$$g_1$ $\wedge$ $\neg$$g_3$)$\Un$$g_2$)) $\wedge$
G($g_2$ $\rightarrow$ X(($\neg$$g_1$ $\wedge$ $\neg$$g_2$)$\Un$$g_3$)).
Now, we can informally describe the mission. The mission is ``Always gather data from g3, g1, g2 in this order and upload the collected data to $u_1$ and $u_2$. Once data gathering is finished, do not visit gather locations until the data is uploaded. Once uploading is finished, do not visit upload locations until gathering data. You should always avoid the road from $i_4$ to $i_2$ when you head to $u_1$ from $g_1$ and the road from $i_1$ to $i_2$ when you head to $u_2$ from $g_2$''. The following formula represents this mission:
\begin{center}
$\phi_{single}$ := $\phi_O \wedge \phi_G \wedge \phi_U \wedge \psi_1 \wedge \psi_2 \wedge$ GF($\pi$).
\end{center}

Assume that initially, the robot is in $i_3$ and all nodes are final nodes.
When we made a cross product with the road and the specification, we could get 36824 states, 350114 edges, and 450 final states.
Not removing some atomic propositions, the specification was not satisfiable.
AAMRP took 15 min 34.572 seconds, and suggested removing $g_3$.
Since the original specification has many $g_3$ in it, we had to trace which $g_3$ from the specification should be removed. Hence, we revised the LTL2BA (\cite{GastinO01cav}), indexing each atomic proposition on the transitions and states (see \citealt{LTL2BA_CPSLAB}).%
Two $g_3$ are mappped to the same transition on the specification automaton in   
($\neg$$g_1$ $\wedge$ $\neg$$g_2$)$\Un$$g_3$ of $\phi_O$ and in $\varphi_g := g_1 \vee g_2 \vee g_3$ in $\phi_U$.
\exmend
\end{exmp}

The last example shows somewhat different missions with multiple robots. If the robots execute the gather and upload mission, persistently, we could assume that the battery in the robots should be recharged.

\begin{figure}[t]
\centering
\includegraphics[width=7cm]{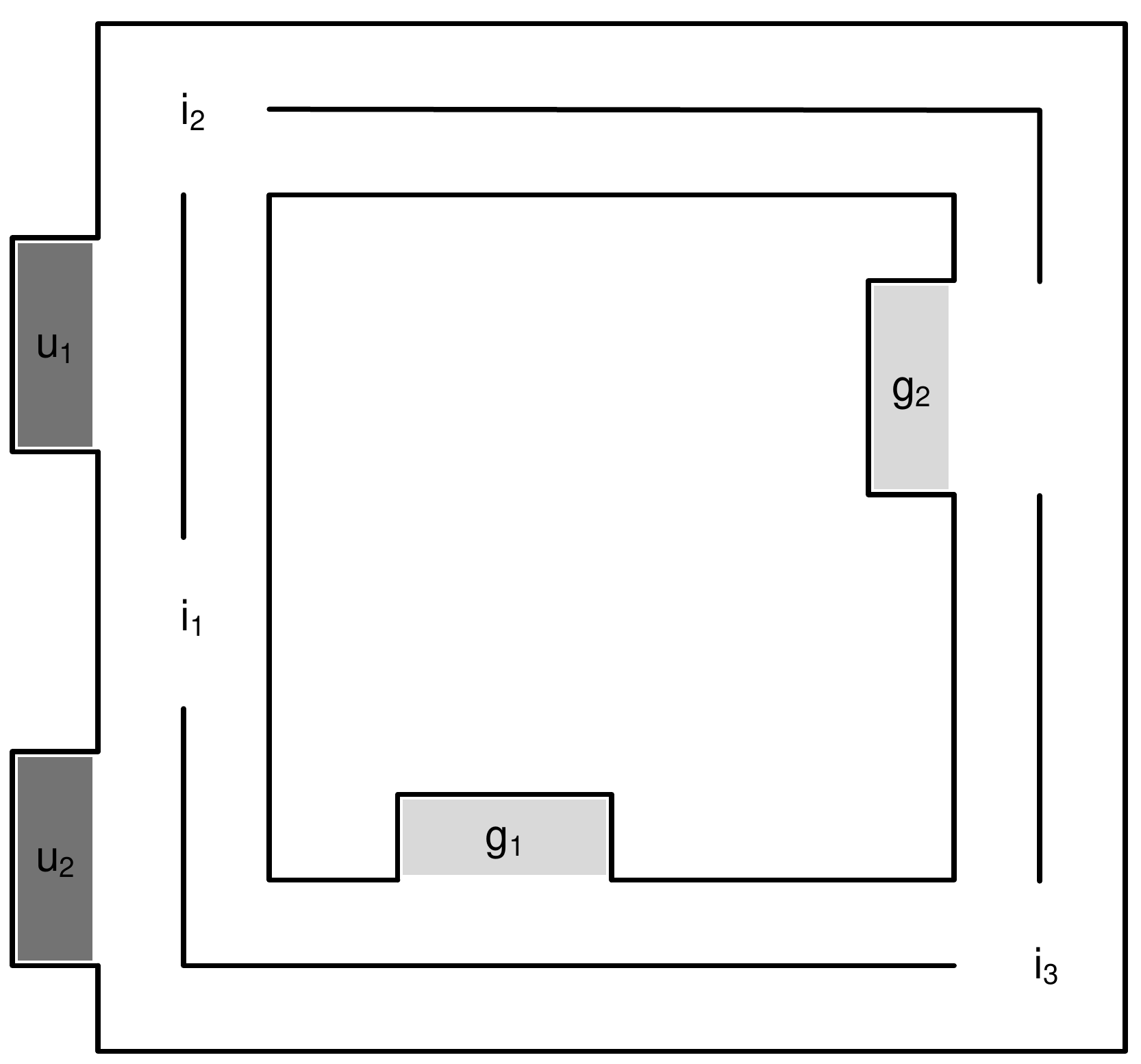}
\caption{Schematic illustration of the simple road network environment of Example \ref{exm:charging_mission}. The robots can stay upload locations $u_1$ and $u_2$ to recharge the battery.}
\label{pic:multi_robot_road}
\end{figure}

\begin{exmp}[Charging while Uploading]
\label{exm:charging_mission}
In this exaple, we assume that robots can recharge their battery in upload locations so that robots are reqired to stay at the upload locations as much as possible.
We also assume that each gathering localtion has a dedicated upload location such that $g_1$ has $u_1$ as an upload location, and $g_2$ has $u_2$ as an upload location.
For this example, we revised the road network so that we remove the gather location $g_3$ and the intersection $i_4$ to make the network simpler for this mission.
We also positioned the upload locations next to each other.
We assume that the power source is shared and it has just two charging statations (see in Fig. 12).
We can describe the mission as follows:
``Once $robot_1$ finishes gathering data at $g_1$, $robot_1$ should not visit the gather locations until the data is uploaded at $u_1$.
Once $robot_2$ fisniehs gathering data at $g_2$, $robot_2$ shoud not visit the gather locations until the data is uploaded at $u_2$.
Once the data is uploaded at $u_1$ or $u_2$, $robot_1$ or $robot_2$ should stay there until a gather locaiton is not occupied.
Persistently, gather data from $g_1$ and $g_2$, avoiding the road from $g_2$ to $i_2$.''
The following formula represents this mission: 


\begin{center}
  $\phi_{charging}$ := G($g_{11}$ $\rightarrow$ X($\neg$$g_{11}$ $\wedge$ $\neg$$g_{21}$) $\Un$$u_{11}$) $\wedge$\\
  G($g_{22}$ $\rightarrow$ X($\neg$$g_{22}$ $\wedge$ $\neg$$g_{12}$) $\Un$$u_{22}$) $\wedge$\\
  G($u_{11}$ $\rightarrow$ $u_{11}$ $\Un$ $\neg$$g_{22}$) $\wedge$\\
  G($u_{22}$ $\rightarrow$ $u_{22}$ $\Un$ $\neg$$g_{11}$) $\wedge$\\
  GF$g_{11}$ $\wedge$ GF$g_{22}$ $\wedge$\\
  G$\neg$($g_{21}$ $\wedge$ X$i_{21}$) $\wedge$\\
  G$\neg$($g_{22}$ $\wedge$ X$i_{22}$).
\end{center}

Assume that initially, $robot_1$ is in $i_1$, $robot_2$ is in $i_2$, and all nodes are final nodes.
From the cross product with the road and the specification, there was 65966 states, 253882 transitions, and 504 final nodes.
For this example, we computed a synchronized environtment for two robots, and in this environment, atomic propositions were duplicationed for each robot.
For example, a gather location $g_1$ is duplicated to $g_{11}$ for $robot_1$ and $g_{12}$ for $robot_2$.
With this synchronized environment, we could avoid robots to be colliding and to be in the same location at the same time.
However, not removing some atomic propositions, the specification was unsatisfiable.
AAMRP took 24 min 22.578 seconds, and suggested removing $u_{22}$ from $robot_2$.
The two occurances of $u_{22}$ were in G($g_{22}$ $\rightarrow$ X($\neg$$g_{22}$ $\wedge$ $\neg$$g_{12}$) $\Un$$u_{22}$) and in the second $u_{22}$ of G($u_{22}$ $\rightarrow$ $u_{22}$ $\Un$ $\neg$$g_{11}$) as indicated by our modified LTL2BA toolbox.
The suggested path from AAMRP for each robot is as followings:
\begin{center}
$path_{robot_1}$ = $i_{11}$$i_{21}$$u_{11}$$u_{11}$$i_{11}$($i_{21}$$g_{21}$$i_{31}$$g_{11}$$i_{11}$$i_{21}$$u_{11}$$u_{11}$$u_{11}$$u_{11}$$u_{11}$$u_{11}$$u_{11}$$u_{11}$$u_{11}$$i_{11}$)$^+$
\end{center}
\begin{center}
$path_{robot_2}$ = $i_{22}$$u_{12}$$i_{12}$$u_{22}$$u_{22}$($u_{22}$$u_{22}$$u_{22}$$u_{22}$$u_{22}$$u_{22}$$u_{22}$$i_{32}$$g_{12}$$i_{12}$$i_{22}$$g_{22}$$i_{32}$$g_{12}$$i_{12}$$u_{22}$)$^+$
\end{center}

\exmend
\end{exmp}

For the experiments, we utilized the ASU super computing center which consists of clusters of Dual 4-core processors, 16 GB Intel(R) Xeon(R) CPU X5355 @2.66 Ghz. 
Our implementation does not utilize the parallel architecture. 
The clusters were used to run the many different test cases in parallel on a single core. 
The operating system is CentOS release 5.5.

In order to assess the experimental approximation ratio of AAMRP, we compared the solutions returned by AAMRP with the brute-force search.
The brute-force search is guaranteed to return a minimal solution to the MRP problem.

\begin{table*}
\begin{center}
\resizebox{\textwidth}{!}{%
\begin{tabular}{|c|ccc|ccc|c|ccc|ccc|c|ccc|}
\hline
Nodes & \multicolumn{7}{c|} {BRUTE-FORCE SEARCH} & \multicolumn{7}{c|} {AAMRP} & \multicolumn{3}{c|} {RATIO} \\ 
\hline
      & \multicolumn{3}{c|} {TIMES (SEC)} & \multicolumn{3}{c|} {SOLUTIONS (SIZE)}  &       & \multicolumn{3}{c|} {TIMES (SEC)} & \multicolumn{3}{c|} {SOLUTIONS (SIZE)}  &         & \multicolumn{3}{c|} {} \\
\hline
      & min     & avg       & max      & min     & avg       & max     & succ  & min  	& avg   	& max      & min     & avg       & max     & succ    &  min  &    avg   & max   \\
\hline
9     &  0.037  &  0.104   &  1.91   & 1     &  1.97    & 5   & 200/200 &  0.022  & 0.061   & 1.17   & 1    & 1.975  & 5 & 200/200 &     1 & 1.0016 & 1.333 \\
\hline
100   &  0.069  &  510.18  & 20786 & 1     &  3.277   & 13   &  198/200 &  0.038  & 0.076 & 0.179  & 1    & 3.395  & 15 & 200/200 &     1 & 1.0006 & 1.125 \\
\hline
196   &  0.066  &  1025.44 & 25271 & 1     &  3.076   & 8   &  171/200 &  0.007  & 0.188  & 0.333  & 1    & 4.285  & 17 & 200/200 &     1 & 1        & 1 \\
\hline
324   &  0.103  &  992.68 & 25437 & 1     &  2.379   & 6   &  158/200 &  0.129  & 0.669  & 1.591  & 1    & 4.155  & 20 & 200/200 &     1 & 1        & 1.2 \\
\hline
400   &  0.087 &  1110.05 & 17685 & 1     &  2.692   & 6   &  143/200 &  0.15   &  0.669 & 1.591  & 1    & 5  & 24 & 200/200 &     1 & 1        & 1 \\
\hline
529   &  0.14  &  2153.90 & 26895 & 1     &  2.591   & 5   &  137/200 &  0.382  &  1.88    & 4.705  & 1    & 5.115  & 30 & 200/200 &     1 & 1        & 1 \\
\hline
\end{tabular}
}
\end{center}
\caption{Numerical Experiments: Number of nodes versus the results of brute-force search and AAMRP. 
Under the brute-force search and AAMRP columns the numbers indicate computation times in $\sec$.
RATIO indicates the experimentally observed approximation ratio to the optimal solution.} 
\label{experiment_result_brute} 
\end{table*}

\begin{table*}
\begin{center}
\begin{tabular}{|c|ccc|c|}
\hline
Nodes & \multicolumn{4}{c|} {AAMRP} \\ 
\hline 
      & \multicolumn{3}{c|} {TIMES} &    \\ 
\hline 
	  &    min   &    avg   &   max    & succ  \\
\hline
1024   &   0.125  &    0.23  &   0.325  & 9/10  \\
\hline
10000  &  15.723  &   76.164 & 128.471  & 9/10  \\
\hline
20164  &   50.325 &  570.737 & 1009.675 & 8/10  \\
\hline
50176  &  425.362 & 1993.449 & 4013.717 & 3/10  \\
\hline
60025  & 6734.133 & 6917.094 & 7100.055 & 2/10  \\
\hline
\end{tabular}
\end{center}
\caption{Numerical Experiments: Number of nodes versus the results of AAMRP. 
Under the TIMES columns the numbers indicate computation times in $\sec$.}
\label{experiment_result_various_only} 
\end{table*}

We performed a large number of experimental comparisons on random benchmark instances of various sizes.
We used the same instances which were presented in \citealt{KimFS12icra, KimFS12iros}.
The first experiment involved randomly generated DAGs.
Each test case consisted of two randomly generated DAGs which represented an environment and a specification. 
Both graphs have self-loops on their leaves so that a feasible lasso path can be found. 
The number of atomic propositions in each instance was equal to four times the number of nodes in each acyclic graph. 
For example, in the benchmark where the graph had 9 nodes, each DAG had 3 nodes, and the number of atomic propositions was 12. 
The final nodes are chosen randomly and they represent 5\%-40\% of the nodes. 
The number of edges in most instances were 2-3 times more than the number of nodes.

Table \ref{experiment_result_brute} compares the results of the brute-force search with the results of AAMRP on test cases of different sizes (total number of nodes).
For each graph size, we performed 200 tests and we report minimum, average and maximum computation times in second and
minimum, average and maximum numbers of atomic propositions for each instance solution.
AAMRP was able to finish the computation and returned a minimal revision for all the test cases, but brute-force search was not able to finish all the computation within a 8 hours window.

Our brute-force search checks all the combinations of atomic propositions.
For example, given $n$ atomic propositions, it checks at most $2^n$ cases.
It uses breath first search to check the reachability for the prefix and the lasso part.
If it is reachable with the chosen atomic propositions, then it is finished.
If it is not reachable, then it chooses another combination until it is reachable.
Since brute-force search checks all the combinations of atomic propositions, the success mostly depends on the time limit of the test. We remark that the brute-force search was not able to provide an answer to all the test cases within a 8 hours window.
The comparison for the approximation ratio was possible only for the test cases where brute-force search successfully completed the computation.
Note that in the case of 529 Nodes, even though the maximum RATIO is 1, the maximum solution from brute-force does not match with the maximum solution from AAMRP. One is 5 and another is 30.
This is because the number of success from brute-force search is 137 / 200 and only comparing this success with the ones from AAMRP, the maximum RATIO is still 1.

An interesting observation is that the maximum approximation ratio is experimentally determined to be less than 2.
For the randomly generated graphs that we have constructed the bound apppears to be 1.333.
However, as we showed in the example \ref{exmp:heur:03}, it is not easy to construct random examples that produce higher approximation ratios. Such example scenarios must be carefully constructed in advance.

In the second numerical experiment, we attempted to determine the problem sizes that our prototype implementation of AAMRP in Python can handle.
The results are presented in Table \ref{experiment_result_various_only}.
We observe that approximately 60,025 nodes would be the limit of the AAMRP implementation in Python.

%% file: related_work.tex
\section{Related work}

The automatic specification revision problem for automata based planning techniques is a relatively new problem.

A related research problem is query checking \citealt{ChechikG03cav}, \citealt{GurfinkelDC02sigsoft}. In query checking, given a model of the system and a temporal logic formula $\phi$, some subformulas in $\phi$ are replaced with placeholders.
Then, the problem is to determine a set of Boolean formulas such that if these formulas are placed into the placeholders. Then, the problem is to determine a set of Boolean formulas such that if these formulas are placed into the placeholders, then $\phi$ holds on the model.
The problem of revision as defined here is substantially different from query checking.
For one, the user does not know where to position the placeholders in the formula when the planning fails.

The papers \citealt{DingZ05ismis}, \citealt{FingerW08bsai} present an also related problem.
It is the problem of revising a system model such that it satisfies a temporal logic specification.
Along the same lines, one can study the problem of maximally permissive controllers for automata specification \citealt{ThistleW94}.
Note that in this paper, we are trying to solve the opposite problem, i.e., we are trying to relax the specification such that it can be realized on the system.
The main motivation for our work is that the model of the system, i.e., the environment and the system dynamics, cannot be modified and, therefore, we need to understand what we can be achieved with the current constraints.

Finding out why a specification is not satisfiable on a model is a problem that is very related to the problems of {\it vacuity} and {\it coverage} in model checking \citealt{Kupferman2008fmcad}.
Another related problem is the detection of the causes of unrealizability in LTL games.
In this case, a number of heuristics have been developed in order to localize the error and provide meaningful information to the user for debugging \citealt{CimattiRST,KonighoferHB09fmcad}.
Along these lines, LTLMop \citealt{RamanK11cav} was developed to debug unrealizable LTL specifications in reactive planning for robotic applications.
\yhl{\citealt{RamanLFLMK13} also provided an integrated system for non-expert users to control robots for high-level, reactive tasks through natural language. This system gives the user natural language feedback when the original intention is unsatisfiable. \citealt{RamanK13} introduced an approach to analyze unrealizable robot specifications due to environment's limitation. They provide how to find the minimal unsatisfiable cores, such as deadlock and livelock, for propositional encodings, searching for some sequence of states in the environment.}

Over-Subscription Planning (OSP) \citealt{Smith04} and Partial Satisfaction Planning (PSP) \citealt{vandenBriel2004} are also very related problems.
OSP finds an appropriate subset of an over-subscribed, conjunctive goal to meet the limitation of time and energy consumption. PSP explains the planning problem where the goal is regarded as soft constraints and trying to find a good quality plan for a subset of the goals.
OSP and PSP have almost same definition, but there is also a difference.
OSP regards the resource limitations as an important factor of partial goal to be satisfied, while PSP chooses a trade-off between the total action costs and the goal utilities where handling the plan quality.


In \citealt{Gobelbecker2010}, the authors investigated situations in which a planner-based agent cannot find a solution for a given planning task. They provided a formalization of coming up with excuses for not being able to find a plan and determined the computational complexity of finding excuses. On the practical side, they presented a method that is able to find good excuses on robotic application domains.

Another related problem is the Minimum Constraint Removal Problem (MCR) \citealt{Kris12}. MCR concentrates on finding the least set of violating geometric constraints so that satisfaction in the specification can be achieved.

\yhl{In \citealt{CizeljB13}, authors introduced a related problem which is of automatic formula revision for Probabilistic Computational Tree Logic (PCTL) with noisy sensor and actuator. Their proposed approach uses some specification update rules in order to revise the specification formula until the supervisor is satisfied.}
\citealt{TumovaEtAl13acc} is closely related with our work. It takes as input a transition system, and a set of sub-specifications in LTL with each reward, and constructs a strategy maximizing the total reward of satisfiable sub-specifications. If a whole sub-specification is not feasible, then it is discarded. In our case, we try to minimize revising the sub-specification if it is infeasible.
In \citealt{KimF14icra}, we also expended our approach with quantitative preference. While revising the sub-specification, it has two approaches to get revision. Instead of finding minimum number of atomic propositions, it tries to minimize the sum of preference levels of the atomic propositions and to minimize the maximum preference level of the atomic propositions.

%% file: conclusions.tex
\section{Conclusions}
In this paper, we proved that the minimal revision problem for specification automata is NP-complete. We also provided a polynomial time approximation algorithm for the problem of minimal revision of specification automata and established its upper bound for a special case.
Furthermore, we provided examples to demonstrate that an approximation ratio cannot be established for this algorithm.

The minimal revision problem is useful when automata theoretic planning fails and the modification of the environment is not possible. 
In such cases, it is desirable that the user receives feedback from the system on what the system can actually achieve.
The challenge in proposing a new specification automaton is that the new specification should be as close as possible to the initial intent of the user.
Our proposed algorithm experimentally achieves approximation ratio very close to 1.
Furthermore, the running time of our prototype implementation is reasonable enough to be able to handle realistic scenarios.

Future research will proceed along several directions.
Since the initial specification is ultimately provided in some form of natural language, we would like the feedback that we provide to be in a natural language setting as well.
Second, we plan on developing a robust and efficient publicly available implementation of our approximation algorithm.

%% file: appendix.tex




\section{Appendix: NP-completeness of the Minimal Connecting Edge Problem}

\newcommand\CEdge{\hat E}
\newcommand\ParCEdge{\Ec}

We will prove the Minimal Connecting Edge (MCE) problem is NP-Complete.
MCE is a slightly simpler version of the Minimal Accepting Path (MAP) problem and, thus, MAP is NP-Complete as well.

In MCE, we consider a directed graph $G = (E,V)$ with a source $s$ and a sink $t$ where there is no path from $s$ to $t$.
We also have a set of candidate edges $\CEdge$ to be added to $E$ such that the graph becomes connected and there is a path from $s$ to $t$.
Note that if the edges in $\CEdge$ have no dependencies between them, then there exists an algorithm that can solve the problem in polynomial time.
For instance, Dijkstra's algorithm \citealt{CormenLRS01} applied on the weighted directed graph $G = (V,E\cup \CEdge,w)$ where the edges in $\CEdge$ are assigned weight 1 and the edges in E are assigned weight $0$ solves the problem efficiently.

However, in MCE, the set $\CEdge$ is partitioned in a number of classes $\CEdge_1, ..., \CEdge_n$ such that if an edge $e_i$ is added from $\CEdge_i$, then all the other edges in $\CEdge_i$ are added as well to $G$.
This corresponds to the fact that if we remove a predicate from a transition in $\ASPEC$, then a number of transitions on $G_{\Ac}$ are affected.
Let us consider the $G_{\Ac}$ in Fig. \ref{fig:product} as an example.
Here, $e_0$, $e_2$ and $e_4$ correspond to $y((s_1,s_1),\pi_0)$, $e_1$ and $e_5$ to $y((s_1,s_1),\pi_2)$ and $e_3$ to $y((s_1,s_1),\pi_3)$.
Thus, $\{e_0,e_1,e_2,e_3,e_4,e_5\} \in \CEdge$ and there exist three classes $\CEdge_i$, $\CEdge_j$ and $\CEdge_j$ in the partition such that $\{e_0, e_2, e_4\} \subseteq  \CEdge_i$, $\{e_1, e_5\} \subseteq  \CEdge_j$ and $e_3 \in \CEdge_k$.

\ifthenelse {\boolean{BGRAPHPDF}}
{
\begin{figure}[h]
\centering

\includegraphics[width=8cm]{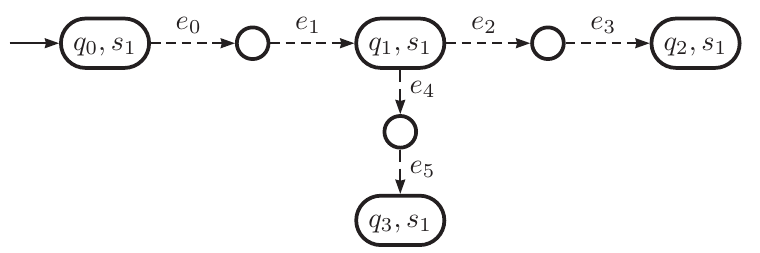}

\caption{The MCE instance that corresponds to $G_{\Ac}$ from Fig. \ref{fig:product}.
The dashed edges denote candidate edges in $\CEdge$.}
\label{fig:graph:mce}
\end{figure}

}
{
\begin{figure}[h]
\centering
\VCDraw{%
\begin{VCPicture}{(-1,-2.4)(9,1.2)}
\ChgStateLabelScale{0.8}
\StateVar[{q_0, s_1}]{(-1,1)}{A} \StateVar[q_1, s_1]{(4,1)}{B} \StateVar[q_2, s_1]{(9,1)}{C} \StateVar[q_3, s_1]{(4,-2)}{D}
\SmallState\StateVar[]{(1.5,1)}{A1} \SmallState\StateVar[]{(6.5,1)}{B1} \SmallState\StateVar[]{(4,-0.5)}{B2}
\Initial{A}
\ChgEdgeLineStyle{dashed} 
\ChgEdgeLabelScale{0.8}
\EdgeL{A}{A1} {e_0}
\EdgeL{A1}{B} {e_1}
\EdgeL{B}{B1} {e_2}
\EdgeL{B1}{C} {e_3}
\EdgeL{B}{B2} {e_4}
\EdgeL{B2}{D} {e_5}
\end{VCPicture}}
\caption{The MCE instance that corresponds to $G_{\Ac}$ from Fig. \ref{fig:product}.
The dashed edges denote candidate edges in $\CEdge$.}
\label{fig:graph:mce}
\end{figure}
}

\begin{prob}[Minimal Connecting Edge (MCE)]
\mathsc{Input:} Let $G=(V,E)$ be a directed graph with a source $s$
and a distinguished sink node $t$.  We assume that there is no path in
$G$ from $s$ to $t$. Let $\CEdge \subseteq V \times V$ be a set such that
$\CEdge \cap E = \emptyset$. We partition $\CEdge$ into $\ParCEdge =
\{\CEdge_1,\ldots,\CEdge_m\}$.  Each edge $e \in \CEdge$ has a weight $W(e) \geq 0$.

\mathsc{Output:} Given a weight limit $W$, determine if there is a
selection of edges $R \subseteq \CEdge$ such that 
\begin{enumerate}
\item there is a path from
$s$ to $t$ in the graph with all edges $E \cup R$, 
\item $\sum_{e \in
  \cup \R} W(e) \leq W$ and 
\item For each $\CEdge_i \in\ParCEdge$, if $\CEdge_i \cap R
\not=\emptyset$ then $\CEdge_i \subseteq R$.
\end{enumerate}
\end{prob}

\begin{thm}
MCE is NP-complete.
\end{thm}

\begin{proof}
The problem is trivially in NP. Given a selection of edges from $\CEdge$,
we can indeed verify that the source and sinks are connected, the
weight limit is respected and that the selection is made up of a union
of sets from the partition.

We now claim that the problem is NP-Complete. We will reduce from
3-CNF-SAT.  Consider an instance of 3-CNF-SAT with variables
$X=\{x_1,\ldots,x_n\}$ and clauses $C_1,\ldots,C_m$. Each clause is a
disjunction of three literals.  We will construct graph $G$ and family
of edges $\ParCEdge$.  The graph $G$ has edges $E$ made up of variable and
clause ``gadgets''.

\paragraph{Variable Gadgets} For each variable $x_i$, we create $6$ nodes
$u_i$, $u_i^{t}$, $v_i^{t}$, $u_i^f$, $v_i^{f}$,  and $v_i$. The gadget is
shown in Fig.~\ref{fig:gadget:var}.  The node $u_i$ is called the
entrance to the gadget and $v_i$ is called the exit. The idea is that if the
variable is assigned true, we will take the path
\[u_i \rightarrow u_i^{t} \rightarrow v_i^{t} \rightarrow v_i\] 
to traverse through the gadget from its entrance to exit. The missing
edge $u_i^{t} \rightarrow v_i^{t}$ will be supplied by one of the edge
sets. If we assign the variable to false, we will instead traverse
\[u_i \rightarrow u_i^{f} \rightarrow v_i^{f} \rightarrow v_i\] 

Variable gadgets are connected to each other in $G$ by adding edges
from $v_1$ to $u_2$, $v_2$ to $u_3$ and so on until $v_{n-1}
\rightarrow u_n$. The node $u_1$ is the source node.

\ifthenelse {\boolean{BGRAPHPDF}}
{
\begin{figure}[ht]
\centering

\includegraphics[width=4.5cm]{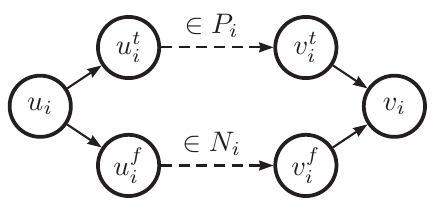}
\caption{ A single variable gadget. 
Solid edges are present in the original graph $G$ that will be constructed.  
Dashed edges $(u_i^t,v_i^t)$ or between $(u_i^f,v_i^f)$ will be supplied by one of the edge sets in $\CEdge$.}
\label{fig:gadget:var}
\end{figure}

}
{
\begin{figure}[ht]
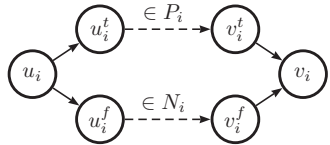

\begin{center}
\VCDraw{%
\begin{VCPicture}{(-1,-1.2)(6,2)}
\ChgStateLabelScale{0.8}
\FixStateDiameter{1.1cm}
\State[u_i]{(-0.5,0)}{A} 
\State[u_i^t]{(1,1)}{B} \State[u_i^f]{(1,-1)}{C} 
\State[v_i^t]{(4,1)}{D} \State[v_i^f]{(4,-1)}{E} 
\State[v_i]{(5.5,0)}{F} 
\EdgeL{A}{B}{} 
\EdgeL{A}{C}{} 
\EdgeL{D}{F}{} 
\EdgeL{E}{F}{} 
\ChgEdgeLineStyle{dashed} 
\ChgEdgeLabelScale{0.8}
\EdgeL{B}{D}{\in P_i} 
\EdgeL{C}{E}{\in N_i} 
\end{VCPicture}}\\
\end{center}
\caption{ A single variable gadget. 
Solid edges are present in the original graph $G$ that will be constructed.  
Dashed edges $(u_i^t,v_i^t)$ or between $(u_i^f,v_i^f)$ will be supplied by one of the edge sets in $\CEdge$.}
\label{fig:gadget:var}
\end{figure}
}

\paragraph{Clause Gadgets} For each clause $C_j$ of the form 
$(\ell_{j1} \lor \ell_{j2} \lor \ell_{j3})$, we add a clause gadget
consisting of eight nodes: entry node $a_j$, exit node $b_j$ and nodes
$a_{j1}, b_{j1}$, $a_{j2},b_{j2}$ and $a_{j3},b_{j3}$ corresponding to
each of the three literals in the clause.  The idea is that a path
from the entry node $a_j$ to exit node $b_j$ will exist if the clause
$C_j$ will be satisfied. Figure~\ref{fig:gadget:cls} shows how the
nodes in a clause gadget are connected.

\ifthenelse {\boolean{BGRAPHPDF}}
{
\begin{figure}[ht]
\centering

\includegraphics[width=4.5cm]{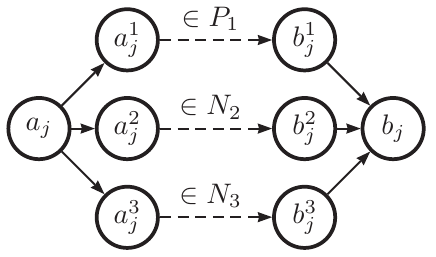}

\caption{The clause gadget for a clause with three literals. 
The clause shown here is $(x_{1}\ \lor\ \overline{x_{2}}\ \lor\ \overline{x_{3}})$. 
The corresponding missing edges will be added to the set $P_{1}, N_{2}, N_{3}$, respectively, as shown in figure. }
\label{fig:gadget:cls}
\end{figure}

}
{
\begin{figure}[ht]
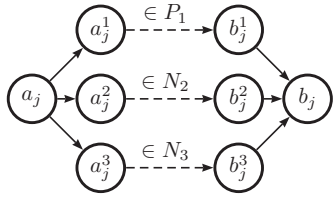

\begin{center}
\VCDraw{%
\begin{VCPicture}{(-1,-1.7)(6,1.7)}
\ChgStateLabelScale{0.8}
\FixStateDiameter{1.1cm}
\State[a_j]{(-0.5,0)}{A} 
\State[a_j^1]{(1,1.5)}{B1} \State[a_j^2]{(1,0)}{B2} \State[a_j^3]{(1,-1.5)}{B3}
\State[b_j^1]{(4,1.5)}{C1} \State[b_j^2]{(4,0)}{C2} \State[b_j^3]{(4,-1.5)}{C3}
\State[b_j]{(5.5,0)}{D} 
\EdgeL{A}{B1}{} 
\EdgeL{A}{B2}{} 
\EdgeL{A}{B3}{} 
\EdgeL{C1}{D}{} 
\EdgeL{C2}{D}{} 
\EdgeL{C3}{D}{} 
\ChgEdgeLineStyle{dashed} 
\ChgEdgeLabelScale{0.8}
\EdgeL{B1}{C1}{\in P_1} 
\EdgeL{B2}{C2}{\in N_2} 
\EdgeL{B3}{C3}{\in N_3} 
\end{VCPicture}}\\
\end{center}
\caption{The clause gadget for a clause with three literals. 
The clause shown here is $(x_{1}\ \lor\ \overline{x_{2}}\ \lor\ \overline{x_{3}})$. 
The corresponding missing edges will be added to the set $P_{1}, N_{2}, N_{3}$, respectively, as shown in figure. }
\label{fig:gadget:cls}
\end{figure}
}

\paragraph{Structure}
We connect $v_n$ the exit of the last variable gadget for variable
$x_n$ to $a_1$, the entrance for first clause gadget.  The sink node
is $b_m$, the exit for the last clause
gadget. Figure~\ref{fig:gadget:all} shows the overall high level
structure of the graph $G$ with variable and clause gadgets.

\ifthenelse {\boolean{BGRAPHPDF}}
{
\begin{figure}[h]
\centering

\includegraphics[width=7cm]{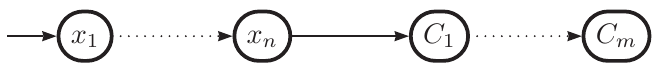}

\caption{Connection between gadgets for variables and clauses.}
\label{fig:gadget:all}
\end{figure}

}
{
\begin{figure}[h]
\centering
\VCDraw{%
\begin{VCPicture}{(0,0.1)(9,-0.1)}
\ChgStateLabelScale{0.8}
\StateVar[x_1]{(0,0)}{A} 
\StateVar[x_n]{(3,0)}{B} 
\StateVar[C_1]{(6,0)}{C} 
\StateVar[C_m]{(9,0)}{D} 
\Initial{A}
\EdgeL{B}{C} {}
\ChgEdgeLineStyle{dotted} 
\EdgeL{A}{B} {}
\EdgeL{C}{D} {}
\end{VCPicture}}
\caption{Connection between gadgets for variables and clauses.}
\label{fig:gadget:all}
\end{figure}
}

\paragraph{Edge Sets}
We design a family $\ParCEdge=\{P_1,\ldots,P_n,N_1,\ldots,N_n\}$.  The set
$P_i$ will correspond to a truth assignment of true to variable $x_i$
and $N_i$ correspond to a truth assignment of false to $x_i$.

$P_i$ has the edge  $(u_i^{t},v_i^{t})$ of weight $1$ and for each
clause $C_j$ containing the literal $x_i$, we add the missing edge
$(a_j^{i},b_j^{i})$ corresponding to this literal in the clause gadget
for $C_j$ to the set $P_i$ with weight $0$.

Similarly, $N_i$ has the edge from $(u_i^f,v_i^f)$ of weight $1$ and
for each clause $C_j$ containing the literal $\overline{x_i}$ it has
the missing edge in the clause gadget for $C_j$ with weight $0$.  We
ask if there is a way to connect the source $u_1$ with the sink $b_m$
with weight limit $\leq n$, where $n$ is the number of variables.

We verify that the sets $P_1,\ldots,P_n, N_1,\ldots,N_n$ partition the
set of missing edges.

\begin{clm}
 If there is a satisfying solution to the problem,
then $u_1$ can be connected to $b_m$ by a choice of edge sets with
total edge weight $\leq n$.
\end{clm}

\begin{proof}
 Take a satisfying solution. If it assigns true to
$x_i$, then choose all edges in $P_i$ else choose all edges $N_i$ if
it assigns false. We claim that this will connect $u_1$ to $b_m$.
First it is clear that since all variables are assigned, it will
connect $u_1$ to $v_n$ by connecting one of the two missing links in
each variable gadget. Corresponding to each clause, $C_j$ there will
be a path from $a_j$ to $b_j$ in the clause gadget for $C_j$. This is
because, at least one of the literals in the clause is satisfied and
the corresponding set $P_i$ or $N_i$ will supply the missing
edge. Furthermore, the weight of the selection will be precisely $n$,
since we add exactly one edge in each variable gadget.
\end{proof}

\begin{clm}
If there is a way to connect source to sink with
weight $\leq n$ then a satisfying assignment exists.
\end{clm}

\begin{proof}
First of all, the total weight for any edge connection
from source to sink is $\geq n$ since we need to connect $u_1$ to
$v_n$ there are $n$ edges missing in any shortest path.  The edges
that will connect have weight $1$, each.  Therefore, if there is a way
to connect source to sink with weight $\leq n$, the total weight must
in fact be $n$. This allows us to conclude that for every variable
gadget precisely one of the missing edges is present. As a result, we
can now form a truth assignment setting $x_i$ to true if $P_i$ is
chosen and false if $N_i$ is. Therefore, the truth assignment will assign
either true to $x_i$ or false and not both thanks to the weight limit
of $n$.
\end{proof}

Next, we prove that each $a_j$ will be connected to $b_j$ in each 
clause gadget corr. to clause $C_j$. Let us assume that this
was using the edge $(a_j^i,b_j^i) \in N_i$. Then, by construction have
that $\overline{x_i}$ was in the clause $C_j$ which is now satisfied
since $N_i$ is chosen, assigning $x_i$ to false.  Similar reasoning
can be used if $(a_j,b_j) \in P_i$. Combining, we conclude that all
clauses are satisfied by our truth assignment.
\end{proof}


%% file: upperboundproof.tex

\section{Appendix: Upper Bound of the Approximation Ratio of AAMRP}
\label{app:bound}

We shall show the upper bound of the approximation algorithm (AAMRP) for a special case.

\begin{thm}
AAMRP on planar Directed Acyclic Graphs (DAG) where all the paths merge on the same node is a polynomial-time 2-approximation algorithm for the Minimal Revision Problem (MRP).
\end{thm}

\begin{proof}

We have already seen that the AAMRP runs in polynomial time.

Let $Y = \{y_1,\ldots,y_m\}$ be a set of Boolean variables
and $G:(V,E)$ be a graph with a labeling function
$L: E \rightarrow \Pc(Y)$, wherein each edge $e \in E$ is labeled
with a set of Boolean variables $L(e) \subseteq Y$.
The label on an edge indicates that the edge is \emph{enabled}
iff all the Boolean variables on the edge are set to true.
Let $v_0 \in V$ be a marked initial state
and $F \subseteq V$ be a set of marked final vertices.

Consider two functions $w':E^* \rightarrow \Pc(Y)$,
and $w:E^* \rightarrow \mathbb{N}$
where $E^*$ represents the set of all finite sequences of edges
of the graph $G$.
Hence, $w'(P)$ for a path $P = \langle v_0,\ldots,v_k \rangle $ is a set of boolean variables
of its constituent edges which makes them enabled on the path $P$:
\[ w'(P) = \bigcup \limits_{i=1}^k {L(v_{i-1},v_i)} \]
while $w(P)$ is the number of the boolean variables of
\[ w(P) = \Bigl \lvert \bigcup \limits_{i=1}^k {L(v_{i-1},v_i)} \Bigr \rvert = \Bigl \lvert w'(P) \Bigr \rvert \]

Given a initial vertex $v_0$, two vertices $v_i$, $v_j$, and a final vertex $v_k$,
let $P_{opt}$ denote the path that produces an optimal revision for the given graph.
Let $P_a$ denote a general revision by {\sc AAMRP}.
Suppose that $P_{opt}$ consists of subpaths $P_{0 i}$, $P_{i j}$, $P_{j k}$,
and $P_a$ consists of subpaths $P_{0 i}$, $P'_{i j}$, $P_{j k}$.

We will discuss the cases when $P_{0 i}$ and $P_{j k}$ are empty later.
The former case can occur when $P_{opt}$ and $P_a$ do not have any common edges from $v_0$ to $v_i$ in the sense that each path takes a different neighbor out of $v_0$.
This case can also occur when $0 < i$ if from $v_0$ to $v_i$ there is no boolean variables to be enabled to make the path activated. 
Likewise, the latter case can occur when $i < j < k$ or when $i \le j = k$.
We do not take $i = j$ unless $j = k$.
Considering both cases together, we can get the possibility that $P_{opt}$ and $P_a$ are entirely different from $v_0$ to $v_k$.

In $v_j$, the AAMRP should relax the weight of the path from $v_0$ to $v_j$, comparing between two paths $P_{i j}$ and $P'_{i j}$.
Thus, we can denote:
\[ w'(P_m) = w'(P_{0 i}) \cup w'(P'_{i j}) \cup w'(P_{j k}), \]
\[ w'(P^*) = w'(P_{0 i}) \cup w'(P_{i j}) \cup w'(P_{j k}). \]

Let $w'(P_{0 i}) \cup w'(P'_{i j}) = \Lambda_a$, $w'(P_{0 i}) \cup w'(P_{i j}) = \Lambda_{opt}$, and $w'(P_{j k}) = \Lambda$. Then, we can denote:
\[ w'(P_a) = \Lambda_a \cup \Lambda, \]
\[ w'(P_{opt}) = \Lambda_{opt} \cup \Lambda. \]
Recall that
\[ w(P_a) = \lvert \Lambda_a \cup \Lambda \rvert, \]
\[ w(P_{opt}) = \lvert \Lambda_{opt} \cup \Lambda \rvert. \]



We will show that  $w(P_a) \le 2w(P_{opt})$.

Note that $w(P_{opt}) \ge 1$, so that $\lvert \Lambda_{opt} \cup \Lambda \rvert \ge 1$. This is because if $w(P_{opt}) = 0$, then it is reachable from $v_0$ to $v_f$ without enabling any boolean variables which are atomic propositions of the specification.

\begin{rem}
\label{rem1}
$\lvert \Lambda_{opt} \cup \Lambda \rvert \ge 1$.
\end{rem}

Note that $\lvert \Lambda_a \rvert \le \lvert \Lambda_{opt} \rvert$. This is because the AAMRP only relaxes the path when it has less number of boolean variables.

\begin{rem}
\label{rem2}
$\lvert \Lambda_a \rvert \le \lvert \Lambda_{opt} \rvert$.
\end{rem}

Note that if $\lvert \Lambda_a \rvert = 0$,
then $\lvert \Lambda_{opt} \cup \Lambda \rvert = \lvert \Lambda_a \cup \Lambda \rvert \le 2\lvert \Lambda_{opt} \cup \Lambda \rvert$.
In this case, 
$\Lambda_a$ is the optimal path if $\Lambda_a = 0$ since
$\Lambda$ is common for the two paths.
I.e., $w(P_a) \le 2w(P_{opt})$.



Consider the case $\lvert \Lambda_a \rvert \ge 1$.
We will prove the claim by contradiction. Assume that $2w(P_{opt}) < w(P_a)$
so that $2\lvert \Lambda_{opt} \cup \Lambda \rvert < \lvert \Lambda_a \cup \Lambda \rvert$.
Let $\lvert \Lambda_a \rvert = \mu$, $\lvert \Lambda_{opt} \rvert = \eta$ and $\lvert \Lambda \rvert = \tau$.
There are four cases.


Case 1: if $\Lambda_{opt} \cap \Lambda = \emptyset$ and $\Lambda_a \cap \Lambda = \emptyset$, then
\[ 2\lvert \Lambda_{opt} \cup \Lambda \rvert < \lvert \Lambda_a \cup \Lambda \rvert \Rightarrow 2(\eta + \tau) < \mu + \tau \]
\[ 2\eta + 2\tau < \mu + \tau \Rightarrow 2\eta + \tau < \mu \]
However, $\mu \le \eta$ by Remark \ref{rem2}. Thus, $\eta + \tau < 0$ which is not possible and it contradicts our assumption.



Case 2: if $\Lambda_{opt} \cap \Lambda \ne \emptyset$ and $\Lambda_a \cap \Lambda = \emptyset$, then
let $\lvert \Lambda_{opt} \cap \Lambda \rvert = \zeta$,
where $1 \le \zeta \le \min(\eta,\tau)$.
\[ 2\lvert \Lambda_{opt} \cup \Lambda \rvert < \lvert \Lambda_a \cup \Lambda \rvert \Rightarrow  2(\eta + \tau - \zeta) < \mu + \tau \]
\[ 2\eta + 2\tau - 2\zeta < \mu + \tau \Rightarrow 2\eta - 2\zeta + \tau < \mu \]

If $\eta \le \tau$, then $\zeta \le \eta$ and $\eta = \zeta + \alpha$, for some $\alpha \ge 0$.
\[ 2(\zeta + \alpha) - 2\zeta + \tau < \mu \Rightarrow 2\alpha + \tau < \mu \]

However, $\eta \le \tau$ and $\mu \le \eta$ by Remark \ref{rem2}. Thus, $2\alpha < 0$ which is not possible and it contradicts our assumption.




If $\eta > \tau$ and $\eta = \tau + \beta$, for some $\beta > 0$, then $\zeta \le \tau$ and $\tau = \zeta + \alpha$, for some $\alpha \ge 0$.
\[ 2(\tau + \beta) - 2\zeta + \tau < \mu \Rightarrow 2\tau - 2\zeta + 2\beta + \tau < \mu \]
\[ 2(\zeta + \alpha) - 2\zeta + 2\beta + \tau < \mu \Rightarrow 2\zeta + 2\alpha - 2\zeta + 2\beta + \tau < \mu \]
\[ 2\alpha + 2\beta + \tau < \mu \Rightarrow 2\alpha + 2\beta + \eta - \beta < \mu \]
\[ 2\alpha + \beta + \eta < \mu \]

However, $\mu \le \eta$ by Remark \ref{rem2}. Thus, $2\alpha + \beta < 0$ which is not possible and it contradicts our assumption.


Case 3: if $\Lambda_{opt} \cap \Lambda = \emptyset$, $\Lambda_a \cap \Lambda \ne \emptyset$, then let $\lvert \Lambda_a \cap \Lambda \rvert = \theta$, where $1 \le \theta \le \min(\mu,\tau)$.
\[ 2\lvert \Lambda_{opt} \cup \Lambda \rvert < \lvert \Lambda_a \cup \Lambda \rvert \Rightarrow 2(\eta + \tau) < \mu + \tau - \theta \]
\[ 2\eta + \tau < \mu - \theta \Rightarrow 2\eta + \tau < \mu \]

However, $\mu \le \eta$ by Remark \ref{rem2}. Thus, $\eta + \tau < 0$ which is not possible and it contradicts our assumption.






Finally the last case: if $\Lambda_{opt} \cap \Lambda \ne \emptyset$, $\Lambda_a \cap \Lambda \ne \emptyset$, then let $\lvert \Lambda_{opt} \cap \Lambda \rvert = \zeta$, where $1 \le \zeta \le \min(\eta,\tau)$, and $\lvert \Lambda_a \cap \Lambda \rvert = \theta$, where $1 \le \theta \le \min(\mu,\tau)$.
\[ 2\lvert \Lambda_{opt} \cup \Lambda \rvert < \lvert \Lambda \cup \Lambda \rvert \Rightarrow 2(\eta + \tau - \zeta) < \mu + \tau - \theta \]
\[ 2\eta + 2\tau - 2\zeta < \mu + \tau - \theta \Rightarrow 2\eta + \tau - 2\zeta < \mu - \theta < \mu \]


If $\eta \le \tau$, $\zeta \le \eta$ and $\eta = \zeta + \alpha$, for some $\alpha \ge 0$, then
\[ 2(\zeta + \alpha) + \tau - 2\zeta < \mu \Rightarrow 2\alpha + \tau < \mu \]

However, $\mu \le \eta$ and $\eta \le \tau$ by Remark \ref{rem2}. Thus, $2\alpha < 0$ which is not possible and it contradicts our assumption.

If $\eta > \tau$, $\eta = \tau + \beta$, for some $\beta > 0$, $\zeta \le \tau$ and $\tau = \zeta + \alpha$, for some $\alpha \ge 0$, then
\[ 2\eta + \tau - 2\zeta < \mu \Rightarrow \eta + (\tau + \beta) + \tau - 2\zeta < \mu \]
\[ \eta + (\zeta + \alpha) + \beta + (\zeta + \beta) -2\zeta < \mu \Rightarrow \eta + 2\zeta + \alpha + 2\beta - 2\zeta < \mu \]
\[ \eta + \alpha + 2\beta < \mu \]

However, $\mu \le \eta$ by Remark \ref{rem2}. Thus, $\alpha + 2\beta < 0$ which is not possible and it contradicts our assumption.

\end{proof}

Therefore, $ \lvert \Lambda_a \cup \Lambda \rvert \le 2\lvert \Lambda_{opt} \cup \Lambda \rvert$, and we can conclude that $w(P_a) \le 2w(P_{opt})$.